\documentclass[review]{elsarticle}
\usepackage[hyphens]{url}
\usepackage{hyperref}
\usepackage[hyphenbreaks]{breakurl}
\usepackage{lineno,hyperref,comment,xcolor}\def\d{{\, \rm d}}
\usepackage{amssymb,amsthm,amsmath,bm}

\usepackage{bm}\usepackage{soul}
\usepackage{empheq}\usepackage{mathtools}
\usepackage{mathrsfs,color}
\usepackage{amsfonts}\usepackage{multirow}
\usepackage{amssymb}\usepackage{caption,comment}
\usepackage{amsmath}\usepackage{float}
\usepackage{amssymb,amsthm}
\usepackage{graphicx,afterpage,cancel}
\usepackage{epstopdf}

\usepackage{soul}
\usepackage{epstopdf}
\graphicspath{ {./Figs/}}
\modulolinenumbers[5]

\newtheorem{theorem}{Theorem}[section]
\newtheorem{proposition}[theorem]{Proposition}

\begin{document}

\begin{frontmatter}

\title{Rigorous Derivation of Stochastic Conceptual Models for the El Ni\~no-Southern Oscillation from a Spatially-Extended Dynamical System}

\author{Nan Chen}
\address{Department of Mathematics, University of Wisconsin-Madison, 480 Lincoln Dr., Madison, WI 53706, USA}
\author{Yinling Zhang\corref{mycorrespondingauthor}}\cortext[mycorrespondingauthor]{Corresponding author}\ead{zhang2447@wisc.edu}
\address{Department of Mathematics, University of Wisconsin-Madison, 480 Lincoln Dr., Madison, WI 53706, USA}

\begin{abstract}
El Ni\~no-Southern Oscillation (ENSO) is the most predominant interannual variability in the tropics, significantly impacting global weather and climate. In this paper, a framework of low-order conceptual models for the ENSO is systematically derived from a spatially-extended stochastic dynamical system with full mathematical rigor. The spatially-extended stochastic dynamical system has a linear, deterministic, and stable dynamical core. It also exploits a simple stochastic process with multiplicative noise to parameterize the intraseasonal wind burst activities. A principal component analysis based on the eigenvalue decomposition method is applied to provide a low-order conceptual model that succeeds in characterizing the large-scale dynamical and non-Gaussian statistical features of the eastern Pacific El Ni\~no events. Despite the low dimensionality, the conceptual modeling framework contains outputs for all the atmosphere, ocean, and sea surface temperature components with detailed spatiotemporal patterns. This contrasts with many existing conceptual models focusing only on a small set of specified state variables. The stochastic versions of many state-of-the-art low-order models, such as the recharge-discharge and the delayed oscillators, become special cases within this framework. The rigorous derivation of such low-order models provides a unique way to connect models with different spatiotemporal complexities. The framework also facilitates understanding the instantaneous and memory effects of stochastic noise in contributing to the large-scale dynamics of the ENSO.
\end{abstract}

\begin{keyword}
Stochastic conceptual model\sep Eigenvalue decomposition\sep 
Spatially-extended stochastic dynamical systems\sep Large-scale ENSO features\sep Multiplicative noise\sep Stochastic discharge-recharge oscillator\sep Stochastic delayed oscillator
\MSC[2010] 35R60\sep 37N10\sep 60H10  \sep 65F15
\end{keyword}

\end{frontmatter}


\section{Introduction}
El Ni\~no-Southern Oscillation (ENSO) is the most predominant interannual variability in the tropics with significant impacts on the global weather and climate \cite{philander1983nino, ropelewski1987global, klein1999remote, mcphaden2006enso}. El Ni\~no and La Ni\~na are opposing patterns, representing the warm and cool phases of a recurring climate pattern across the tropical Pacific. They form an irregular oscillator, shifting back and forth every three to seven years, leading to predictable shifts in sea surface temperature (SST) and disrupting the wind and rainfall patterns across the tropics.

A hierarchy of dynamical and statistical models with different spatiotemporal complexity has been developed to understand the ENSO dynamics across a wide range of spatiotemporal scales. The simplest models are the conceptual or box models \cite{schopf1988vacillations, picaut1996mechanism, jin1997equatorial, wang2004understanding, chen2022multiscale, capotondi2018nature}, which involve a few (stochastic) ordinary differential equations. Some variables in these models are directly related to the commonly used low-dimensional ENSO indicators, such as the Ni\~no 3 index. The models at the next level of complexity are called simple dynamical models \cite{leeuwen2002balanced, fang2018simulating, thual2016simple, chen2017simple}, which consist of a set of partial differential equations to characterize the large-scale dynamics of ENSO. By including more detailed dynamics and coupled relationships with atmosphere and other state variables, the resulting models are named the intermediate complex models (ICMs) \cite{zebiak1987model, battisti1989interannual, neelin1993modes,zhang2003new,chen2022simple}. The ICMs are widely used for studying the mechanisms and the forecast of the ENSO. The ENSO dynamics have also been incorporated into the general circulation models (GCMs) \cite{planton2021evaluating, lau2000impact, guilyardi2020enso, zhao2020effects} to understand its global and regional impacts.

Among these categories of ENSO models, the low-order conceptual models are of particular interest for the following reasons. First, these models are mathematically tractable due to low dimensionality and simple structures. Rigorous analysis of model properties, such as the bifurcations and the stability, can be easily carried out for these models that help understand the mechanisms of the large-scale ENSO features. Second, these models are often utilized to test various hypotheses of physical and causal dependence between different variables. The results can provide valuable guidelines for developing more sophisticated models with refined dynamical structures \cite{jin1997equatorial, jin1997equatorial2}. These conceptual models have also been used as testbeds for guiding simple dynamical models or ICMs to include additional stochastic noise or stochastic parameterizations \cite{chen2022multiscale, chen2022simple}.  Third, because of the low computational cost, these conceptual models can also be used to predict the ENSO indices \cite{latif1998review, xue1994prediction} and study the predictability from a statistical point of view \cite{fang2022quantifying} in light of ensemble forecast methods. Several low-order conceptual models have been independently developed, including the recharge-discharge oscillator \cite{jin1997equatorial, wyrtki1975nino}, the delayed oscillator \cite{suarez1988delayed, battisti1989interannual, mccreary1983model}, the western-Pacific oscillator \cite{weisberg1997western}, and the advective–reflective oscillator \cite{picaut1997advective}. Later, a unified ENSO oscillator motivated by the dynamics and thermodynamics of Zebiak and Cane's coupled ocean–atmosphere model has also been built \cite{wang2001unified}. These models were mainly proposed based on physical intuitions and highlighted one or two specific dynamical features of the ENSO as the building blocks. They have led to many successes in applications.

In this paper, a framework of low-order conceptual models for the ENSO is systematically derived from a spatially-extended stochastic dynamical system with full mathematical rigor. The spatially-extended stochastic dynamical system has a linear, deterministic, and stable dynamical core, describing the interactions between atmosphere, ocean and SST at the interannual time scale. It also exploits a simple stochastic process with state dependent (i.e., multiplicative) noise to parameterize the intraseasonal wind burst activities, including the effect from the Madden-Julian oscillation (MJO), that trigger or terminate the El Ni\~no events. The model captures many observed dynamical features of the ENSO, such as the varying amplitudes and durations of different El Ni\~nos as well as the extreme events. It also reproduces the observed power spectrum and non-Gaussian statistics of the SST in the eastern Pacific, the region of the active ENSO events, thanks to the multiplicative noise. Since the dynamical core is linear, a principal component analysis based on the eigenvalue decomposition method is applied to provide a low-order representation of the large-scale dynamics of the ENSO. By further projecting the stochastic wind bursts to the leading order basis functions, the low-order representation succeeds in reproducing the observed large-scale dynamical and statistical features of the ENSO. In light of such a strategy of reducing the model complexity, a low-order conceptual modeling framework is thus rigorously derived. The stochastic versions of many state-of-the-art low-order oscillators, such as the recharge-discharge and the delayed oscillators, become special cases within this framework.

The development of the low-order conceptual modeling framework here is very different from the unified ENSO oscillator \cite{wang2001unified} and many other models. First, the unified oscillator exploits Zebiak and Cane's model, which utilizes nonlinearity to trigger the ENSO cycles, as a starting reference model to provide guidelines for developing the conceptual model. In contrast, the conceptual model developed here starts from a linear spatially-extended dynamical system, which uses multiplicative noise to induce the ENSO events. Second, instead of providing only a heuristic building block, the linear dynamical core allows rigorous mathematical derivations from the spatially-extended dynamical system to the low-order conceptual model. This procedure facilitates rigorous analysis and direct intercomparison between the two systems. Third, the low-order conceptual modeling framework here provides outputs for all the atmosphere, ocean, and SST components with detailed spatiotemporal patterns. This contrasts with many existing conceptual models, which focus only on a small set of specified state variables and often involve the averaged quantities over a large domain as a coarse-graining representation. Despite the intrinsic low dimensionality of the model, the state variables for atmosphere, ocean, and SST in this conceptual modeling framework are connected via the eigenvectors. Therefore, depending on the quantity of interest, the model can explicitly characterize the relationships between different atmosphere-ocean-SST components contributing to the ENSO dynamics.

The rest of the paper is organized as follows. Section \ref{Sec:Observations} describes the observational data. Section \ref{Sec:PDEModel} presents the simple spatially-extended stochastic dynamical model and its properties. Section \ref{Sec:ROP} utilizes the eigenvalue decomposition to illustrate the reduced order representation of the ENSO dynamics. Section \ref{Sec:CM} shows the mathematical expressions of the conceptual models, including the recharge-discharge and the delayed oscillators. The paper is concluded in Section \ref{Sec:Conclusion}.

\section{Observational Data}\label{Sec:Observations}
The following observational data are utilized in this study. Daily sea surface temperature (SST) data comes from the OISST reanalysis \cite{reynolds2007daily} (\url{https://www.ncdc.noaa.gov/oisst}). Monthly thermocline depth is from the NCEP GODAS reanalysis \cite{behringer2004evaluation} (\url{http://www.esrl.noaa.gov/psd/}). Note that thermocline depth is computed from potential temperature as the depth of the $20^o$C isotherm. Daily zonal winds at 850 hPa are from the NCEP–NCAR reanalysis \cite{kalnay1996ncep} (\url{http://www.esrl.noaa.gov/psd/}). All datasets are averaged meridionally within $5^o$S-$5^o$N in the tropical Pacific ($120^o$E–$80^o$W). Only the anomalies are presented in this study, calculated by removing the monthly mean climatology of the whole period. The eastern and western Pacific regions are defined as $120^o$E-$160^o$W and $160^o$W-$80^o$W, respectively. The Ni\~no 3 SST index is the average of SST anomalies over the region of $150^o$W-$90^o$W. Denote by $T_E$ the time series of SST averaged over the eastern Pacific. The time series $T_E$ and Ni\~no 3 SST index are almost the same as each other.

\section{The Simple Spatially-Extended Stochastic Model}\label{Sec:PDEModel}
\subsection{Review of the model}
The simple spatially-extended stochastic model was originally developed in \cite{thual2016simple}. The starting system is a coupled atmosphere-ocean-SST model, which is given by

\noindent Atmosphere:
\begin{equation}\label{Atmosphere_model_Starting}
\begin{split}
&-yv-\partial_{x}\theta=0\\
&yu-\partial_{y}\theta=0\\
&-(\partial_{x}u+\partial_{y}v)=E_{q}/(1-\overline{Q})
\end{split}
\end{equation}
\noindent Ocean:
\begin{equation}\label{Ocean_model_Starting}
\begin{split}
&\partial_{t}U-c_{1}YV+c_{1}\partial_{x}H=c_{1}{ \tau}\\
&YU+\partial_{Y}H=0\\
&\partial_{t}H+c_{1}(\partial_{x}U+\partial_{Y}V)=0
\end{split}
\end{equation}
\noindent SST:
\begin{equation}\label{SST_model_Starting}
\begin{split}
&\partial_{t}T =-c_1\zeta E_{q}+c_1\eta H.
\end{split}
\end{equation}
The above coupled system consists of a non-dissipative Matsuno–Gill type atmosphere model \cite{matsuno1966quasi, gill1980some}, a simple shallow-water ocean model \cite{vallis2016geophysical} and an SST budget equation \cite{jin1997equatorial2}. The non-dissipative atmosphere is also consistent with the skeleton model for the MJO in the tropics \cite{majda2009skeleton}. In \eqref{Atmosphere_model_Starting}--\eqref{SST_model_Starting}, $u$ and $v$ are the zonal and meridional wind speeds, $\theta$ is the potential temperature, $U$ and $V$ are the zonal and meridional ocean currents, $H$ is the thermocline depth, $T$ is the SST, $E_q=\alpha_q T$ is the latent heat, and $\tau =\gamma u$ is the wind stress along the zonal direction. All the variables are the anomalies from the climatology mean states. There parameters $c_1$, $\overline{Q}$ and $\zeta$ are all constants.
The parameter $c_1$ is a non-dimensional number, representing the ratio of ocean and atmosphere phase speeds over the Froude number, $\overline{Q}$ is related to the background moisture, and $\zeta$ is the latent heating exchange coefficient. The thermocline feedback $\eta(x)$ is a fixed spatial dependent function stronger in the eastern Pacific due to the shallower thermocline. See Panel (c) of Figure \ref{Truth_Simulation}. Also, see the Appendix for a summary of model parameters. Note that different axes in the meridional directions are utilized for the atmosphere ($y$) and ocean ($Y$). This is because the deformation radii of the atmosphere and the ocean are different. The atmosphere covers the entire equatorial band from $[0, L_A]$, and therefore the periodic boundary condition in the zonal direction is applied, namely $u(0,y,t)=u(L_A,y,t)$. The ocean covers the equatorial Pacific $[0, L_O]$ with the boundary conditions $\int_{-\infty}^{\infty}U(0,Y,t)\d Y=0$  and $U(L_O,Y,t)=0$ \cite{cane1981response, jin1997equatorial2}.

Since the most predominant behavior of the ENSO is around the equator and along the zonal direction, it is natural to apply a meridional expansion, which serves as the separation of variables, and then implement a meridional truncation to the first basis in both atmosphere $\phi_0(Y)$ and ocean $\psi_0(y)$ to reduce the complexity of the system. The basis functions in the meridional direction are given by the parabolic cylinder functions \cite{majda2003introduction}. See Panels (a)--(b) in Figure  \ref{Truth_Simulation} for the meridional basis functions, where the first one has a Gaussian profile. The meridional truncation triggers atmosphere Kelvin, Rossby waves $K_A, R_A$, and ocean Kelvin, Rossby waves $K_O, R_O$. Therefore, after removing the meridional dependence from \eqref{Atmosphere_model_Starting}--\eqref{SST_model_Starting}, the coupled system (with dependence only on time $t$ and zonal coordinate $x$) yields.

\noindent Atmosphere:
\begin{equation}\label{Atmosphere_model}
\begin{split}
&\partial_x K_A  = -\chi_A E_q(2-2\bar{Q})^{-1}\\
  &-\partial_x R_A/3  = -\chi_A E_q(3-3\bar{Q})^{-1}\\
  (B.C.)~~~&K_A(0,t) = K_A(L_A,t)\\
(B.C.)~~~ &R_A(0,t) = R_A(L_A,t)
\end{split}
\end{equation}
\noindent Ocean:
\begin{equation}\label{Ocean_model}
\begin{split}
&\partial_t K_O + c_1\partial_x K_O  = \chi_O c_1\tau/2\\
 &\partial_t R_O - (c_1/3)\partial_x R_O  = -\chi_O c_1\tau/3\\
 (B.C.)~~~&K_O(0,t) = r_WR_O(0,t)\\
(B.C.)~~~&R_O(L_O,t) = r_EK_O(L_O,t)
\end{split}
\end{equation}
\noindent SST:
\begin{equation}\label{SST_model}
\begin{split}
 &\partial_t T = - c_1\zeta E_q+ c_1\eta (K_O + R_O).
\end{split}
\end{equation}
The boundary conditions (B.C.) in \eqref{Atmosphere_model}--\eqref{SST_model} are derived from those related to \eqref{Atmosphere_model_Starting}--\eqref{SST_model_Starting}. Periodic boundary conditions are remained for atmosphere waves, while those for the ocean velocity become the reflection boundary conditions for ocean waves. The constants $\chi_A$ and $\chi_O$ are the meridional projection coefficients with $\chi_A = \int_{-\infty}^{\infty}\phi_0(y)\phi_0(y/\sqrt{c})dy$ and $\chi_O = \int_{-\infty}^{\infty}\psi_0(Y)\psi_0(\sqrt{c}Y)dY$, where $c$ is the ratio of ocean and atmosphere phase speed. Once these waves are solved, the physical variables can be reconstructed,
\begin{equation}\label{physical_reconstruction}
\begin{split}
&u = (K_A - R_A)\phi_0 + (R_A/\sqrt{2})\phi_2\\
&\theta = -(K_A + R_A)\phi_0 - (R_A/\sqrt{2})\phi_2\\
&U = (K_O - R_O)\psi_0 + (R_O/\sqrt{2})\psi_2\\
&H = (K_O + R_O)\psi_0 + (R_O/\sqrt{2})\psi_2
\end{split}
\end{equation}
where $\phi_2$ and $\psi_2$ are the third meridional bases of atmosphere and ocean, respectively, resulting from the meridional expansion of the parabolic cylinder function \cite{majda2003introduction}. See Panels (a)--(b) in Figure \ref{Truth_Simulation} for the structures of these bases.

In addition to the interannual state variables, the intraseasonal variabilities are also part of the indispensable ENSO dynamics. Although these intraseasonal variabilities behave like random noise in the interannual time scale, they play vital roles in generating the ENSO events and increasing the ENSO complexity. In particular, atmosphere wind bursts in the intraseasonal time scale, such as the westerly wind bursts (WWBs) \cite{harrison1997westerly, vecchi2000tropical, tziperman2007quantifying}, easterly wind bursts (EWBs) \cite{hu2016exceptionally, levine2016july} and the MJO \cite{hendon2007seasonal, puy2016modulation}, are essential in triggering and terminating the El Ni\~no events. Due to the fast dynamics of wind activities, it is natural to adopt a simple stochastic process to effectively characterize their time evolutions. With the stochastic wind bursts, the zonal wind stress $\tau$ now has two components, $\tau =\gamma(u+u_p)$, where $u$ is directly from the atmosphere model \eqref{Atmosphere_model} while $u_p$ is the contribution from the stochastic wind bursts. Here,
\begin{equation}\label{up_equation}
  u_{p}=a_{p}(t)s_{p}(x)\phi_0(y),
\end{equation}
where for simplicity $s_p(x)$ is a fixed spatial basis function being localized in the western Pacific (see Panel (c) of Figure \ref{Truth_Simulation}), since most of the active wind bursts are observed there. The stochastic amplitude $a_{p}$ is given by the following stochastic differential equation
\begin{equation}\label{ap_equation}
  \frac{\d a_{p}}{\d t}=-d_{p}a_{p}+{\sigma_{p}(T_{W})}\dot{W}(t),
\end{equation}
where $\dot{W}(t)$ is a standard white noise \cite{gardiner1985handbook} and $d_{p}$ is a damping term characterizing the temporal correlation of the wind bursts at the intraseasonal time scale. 

Note that the noise coefficient ${\sigma_{p}(T_{W})}$ is state-dependent (namely, the multiplicative noise) as a function of the averaged SST in the western Pacific.
This is because a warmer SST induces more convective activities, which in turn triggers more wind bursts \cite{vecchi2000tropical, harrison1997westerly}. Here, the state dependence of ${\sigma_{p}(T_{W})}$ is given by a two-state Markov jump process with one active state and one quiescent state:
\begin{equation}
\sigma_p\left(T_W\right)=\left\{\begin{array}{l}
\sigma_{p 0} = 0.2\quad \text { for quiescent state } 0  \\
\sigma_{p 1} = 2.6\quad \text { for active state } 1
\end{array}\right.    
\end{equation}
The transition rate from the quiescent to the active state $\mu_{01}$ and that from the active to the quiescent state $\mu_{10}$ are shown in Panel (d) of Figure \ref{Truth_Simulation}. The explicit equations of them are in the following:
\begin{equation}
    \begin{split}
        & \text{State 1 to 0:}\ \mu_{10}=\frac{1}{4}\left(1-\tanh \left(2 T_W\right)\right)\\
        & \text{State 0 to 1:}\ \mu_{01}=\frac{1}{8}\left(\tanh \left(2 T_W\right)+1\right)
    \end{split}
\end{equation}
Here, transition rates $\mu_{01}$ and $\mu_{10}$ are correlated and anti-correlated with $T_W$, respectively, which reflect the nature of the multiplicative noise. Note that the wind burst activity from the model does not favor specifically westerly or easterly. The WWBs and EWBs generated from the model are entirely random. Only the amplitude of the wind bursts is determined by the transitions related to the SST. The two-state Markov jump process for parameterizing the wind activities is utilized here for simplicity, categorizing the winds by only the active and quiescent phases. A continuous dependence of the wind strength on the SST can easily be incorporated, which leads to a similar SST behavior in the ENSO dynamics \cite{levine2010noise, chen2021bayesian}.

As a final remark, the coupled model \eqref{Atmosphere_model}--\eqref{ap_equation} focuses on the eastern Pacific El Ni\~no events, which is the primary interest of this study. The model does not include the mechanism of generating the central Pacific events \cite{ashok2007nino, lee2010increasing, capotondi2015understanding}. Incorporating the central Pacific events into the model will be briefly discussed in Section \ref{Sec:Conclusion}.

\subsection{Model properties}
The dynamical core of the system \eqref{Atmosphere_model}--\eqref{SST_model} is deterministic, linear and stable. In other words, the solution is a linear damped regular oscillator. The random wind bursts serve as external forcing to trigger the El Ni\~no events and introduce irregularity in the solution. The subsequent La Ni\~na is a consequence of the decaying phase of the ENSO cycle.

Panel (e) in Figure \ref{Truth_Simulation} shows one simulation of the coupled model \eqref{Atmosphere_model}--\eqref{ap_equation} for different fields. The model can reproduce the quasi-regular oscillation behavior of the SST anomaly. When the SST in the western Pacific increases, the wind activities become more significant. If the strong wind activity is dominated by the WWBs, then El Ni\~no events are likely to be triggered. During this process, the center of the atmosphere wind convergence moves towards the east, and the strong eastward ocean zonal current is observed in the western Pacific. The thermocline depth becomes deeper in the eastern Pacific and induces the El Ni\~no events. The model can generate not only the moderate El Ni\~no events, but also the extreme events, known as the super El Ni\~nos, that have strong amplitudes (e.g., the event at $t=490$). In addition to reproducing the super El Ni\~no that mimics the observed 1997-1998 event (e.g., at $t=490$), the model succeeds in generating the so-called delayed super El Ni\~no \cite{thual2019statistical}, similar to the observed 2014-2016 event \cite{chen2017formation, levine2016july}. One such example is the event around $t=474$-$476$, which is due to the peculiar wind structures. A series of WWBs induce the onset of an El Ni\~no event. However, a strong EWB occurs immediately to present the development of such an event to a super El Ni\~no. After a year, another period of strong WWBs happen to trigger the extreme event.

Averaging the SST over the eastern Pacific, the resulting time series is denoted by $T_E$. The model can reproduce very similar non-Gaussian probability density function (PDF) of $T_E$ as the observations \cite{thual2016simple}. This non-Gaussian feature is reproduced with the contribution from the multiplicative noise, without which the model is linear and the associated PDF is Gaussian. Likewise, the spectrum of the model has the most significant power within the band of 3 to 7 years, which is also consistent with observations.

\begin{figure}[ht]
\hspace*{-0cm}\includegraphics[width=1.0\textwidth]{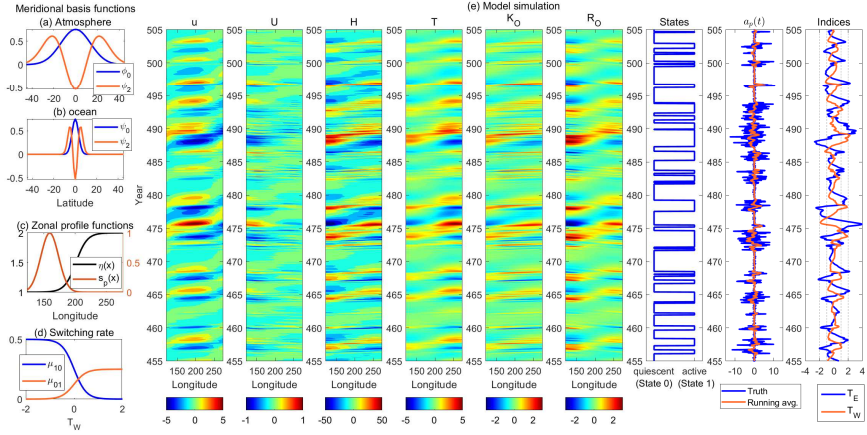}
\caption{Structure and simulation of the spatially-extended stochastic system. Panels (a)--(b): meridional basis functions of the atmosphere ($\phi_0,\phi_2$) and ocean ($\psi_0,\psi_2$). Panel (c): Zonal profile of the thermocline depth $\eta(x)$ and wind burst region $s_p(x)$. Panel (d): the transition rates in the two-state Markov jump process. Panel (e): a random realization of the model simulation. The variables are the atmosphere zonal wind velocity $u$ (unit: m/s), the ocean zonal current $U$ (unit: m/s), the thermocline depth $H$ (unit: m), SST $T$ (unit: $^o$C), ocean Kelvin wave $K_O$ (no unit), ocean Rossby wave $R_O$ (no unit), the state of the Markov jump process, the wind burst time series $a_p(t)$ (unit: m/s), and the two SST indices $T_W$ and $T_E$ (unit: $^o$C), which are the averaged value over the western and eastern Pacific, respectively. In the panel of the wind burst amplitude $a_p(t)$, a 90-day running average of $a_p(t)$ (brown curve) is included to represent the lower frequency part of the wind bursts. }\label{Truth_Simulation}
\end{figure}

\section{Reduced-Order Representation of the ENSO dynamics}\label{Sec:ROP}
\subsection{Discretization of the deterministic, linear and stable dynamics}
Without the stochastic wind bursts, the coupled atmosphere-ocean-SST model \eqref{Atmosphere_model}--\eqref{SST_model} is a linear, deterministic and stable model. Since $u$ (and therefore the wind stress $\tau$) is a function of $K_A$ and $R_A$, the coupled system can be rewritten as
\begin{equation}\label{Combined_system}
\begin{aligned}
    \partial_{t}{ K_O} + c_1\partial_x{ K_O} &= \frac{a}{2}\ ({ K_A} - { R_A}), &&K_O(0,t) = r_WR_O(0,t),\\
    \partial_{t}{ R_O} - \frac{c_1}{3}\partial_x{ R_O} &= -\frac{a}{3}\ ({ K_A} - { R_A}), &&R_O(L_O,t) = r_EK_O(L_O,t),\\
    \partial_{t}{ T}\qquad\qquad\quad\  &=-b{ T} + c_1\eta ({ K_O} + { R_O}),\\
    d_A{ K_A}+\partial_x{ K_A} & = \frac{m_1}{\alpha_q} (\mathbf{1}_{[0,L_O]}\alpha_q { T}-\langle { E_q}\rangle), &&K_A(0,t) = K_A(L_A,t),\\
    d_A{ R_A}-\frac{1}{3}\partial_x{ R_A} & = \frac{m_2}{\alpha_q} (\mathbf{1}_{[0,L_O]}\alpha_q { T}-\langle { E_q}\rangle), && R_A(0,t) = R_A(L_A,t),
\end{aligned}
\end{equation}
where a small damping $d_A=10^{-8}$ is added to guarantee the solution is unique \cite{thual2016simple}, and the new constants introduced in \eqref{Combined_system} are the following ones:
\begin{equation*}
    a = \chi_O c_1 \gamma, \quad b = c_1\zeta\alpha_q, \quad m_1 = -\chi_A \alpha_q/(2 - 2\bar{Q}),\quad m_2 = -\chi_A \alpha_q/(3 - 3\bar{Q}).
\end{equation*}
The domain of $K_O, R_O$ and $T$ is the equatorial Pacific ocean $[0, L_O]$ while the domain of $K_A$ and $R_A$ is the entire equatorial band $[0, L_A]$.

Define a vector $\mathbf{u}$ that contains all the prognostic variables at discrete equal-partitioned grid points,
\begin{equation}\label{define_u}
    \mathbf{u} = (\mathbf{K_O}; \mathbf{R_O}; \mathbf{T}) = (K_{O,1}, \ldots K_{O,{N_O}}, R_{O,1},\ldots R_{O,{N_O}}, T_1,\ldots T_{N_O})^\mathtt{T},
\end{equation}
where $N_O$ is the total number of grid points over the equatorial Pacific region $[0, L_O]$ and $\cdot^\mathtt{T}$ is the vector transpose.
Using an upwind scheme, the discrete form of the linear system \eqref{Combined_system} is given by:
\begin{equation}\label{LinearSystem}
    \frac{\d\mathbf{u}}{\d t} = \mathbf{Mu}.
\end{equation}
The detailed form of $\mathbf{M}$ is included in the Appendix.
With the stochastic wind burst, the above equation can be formally written as:
\begin{equation}\label{LinearSystem_Stochastic}
    \frac{\d\mathbf{u}}{\d t} = \mathbf{Mu} + \mathbf{F},
\end{equation}
where $\mathbf{F}$ is associated with the stochastic forcing. Denote by $\mathbf{s}_\mathbf{u} = (\chi_O c_1\gamma s_p/2,\\
-\chi_O c_1\gamma s_p/3,\mathbf{0})^\mathtt{T}$ and then $\mathbf{F}=\mathbf{s}_\mathbf{u}a_p$, where $s_p$ is written in the discrete form at grid points $1,\ldots, N_O$ and $\mathbf{0}$ is a $1\times N_O$ vector with all zero entries. This is the wind burst forcing acting on the three components $\mathbf{u}=(\mathbf{K}_O,\mathbf{R}_O, \mathbf{T})$ (ocean Kelvin wave, ocean Rossby wave, and SST) of the state variables. See Panel (a) of Figure \ref{z_and_su} for the profile of $\mathbf{s}_\mathbf{u}$. As expected, a positive response is on the ocean Kelvin waves and a negative one is on the ocean Rossby waves, both in the western Pacific. The forcing is not directly imposed on SST; therefore, the third segment of $\mathbf{s}_\mathbf{u}$, corresponding to SST components, is zero.

\begin{figure}[ht]
\hspace*{-0cm}\includegraphics[width=1.0\textwidth]{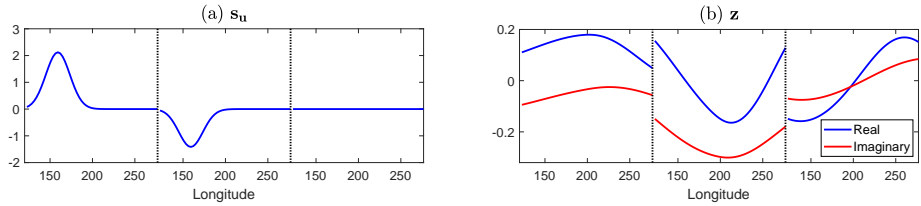}
\caption{Zonal structure of the functions $\mathbf{s}_\mathbf{u}$ and $\mathbf{z}$. The three segments correspond to the three components $\mathbf{u}=(\mathbf{K}_O,\mathbf{R}_O, \mathbf{T})$ (ocean Kelvin wave, ocean Rossby wave, and SST). }\label{z_and_su}
\end{figure}

\subsection{The eigenvalues of the system}
It can be numerically validated that the $3N_O$ eigenvalues all have distinct values and are non-zero. In addition, the $i$-th left eigenvector and the $j$-th right eigenvector of the matrix $\mathbf{M}$ are orthogonal with each other for $i\neq j$. Therefore, in light of the matrix similarity, the following result is arrived at:
\begin{equation}\label{MLambda0}
  \mathbf{X}^{-1}\mathbf{M}\mathbf{X} =\mathbf{\Lambda},
\end{equation}
or equivalently
\begin{equation}\label{MLambda}
  \mathbf{M} = \mathbf{X}\mathbf{\Lambda}\mathbf{X}^{-1},
\end{equation}
where $\mathbf{\Lambda}$ is a diagonal matrix with diagonal entries being eigenvalues and $\mathbf{X}$ contains the $3N_O$ right eigenvectors. Since the right eigenvectors $\mathbf{X}$ are complete, the matrix inverse exists.

Panel (a) of Figure \ref{LinearSolution} shows the eigenvalues associated with the $3N_O\times3N_O$ matrix $\mathbf{M}$, where $N_O=56$ is utilized here and in the subsequent numerical simulations. These eigenvalues are complex numbers. The real and imaginary parts are shown in the y- and x-axis, representing the growth rate and the frequency, respectively. Since the starting system is stable, the growth rates of all the eigenvalues are negative, representing the decay of the solution. These eigenvalues are ranked according to the growth rate. Therefore, the leading two eigenmodes (marked by red dots) are the dominant modes of the system since they have the slowest decaying rates. These two modes appear in pairs; both the eigenvalues and the eigenvectors are complex conjugates. Their eigenvalues have a frequency of $4.6$ years and a decay rate of $1.5$ years, which are consistent with the observed features of ENSO. The reconstructed spatiotemporal patterns associated with these two eigenmodes (by setting the growth rate to be zero for the plotting purpose), as shown in Panel (b) of \eqref{LinearSolution}, mimic the large-scale ENSO oscillation solution. Notably, the leading two eigenmodes are very robust with respect to the spatial resolution $N_O$ as long as $N_O$ is not too small. Therefore, the analysis here provides a theoretical justification for developing a low-order model based on these leading two eigenmodes.

\begin{figure}[ht]
\hspace*{-0cm}\includegraphics[width=1.0\textwidth]{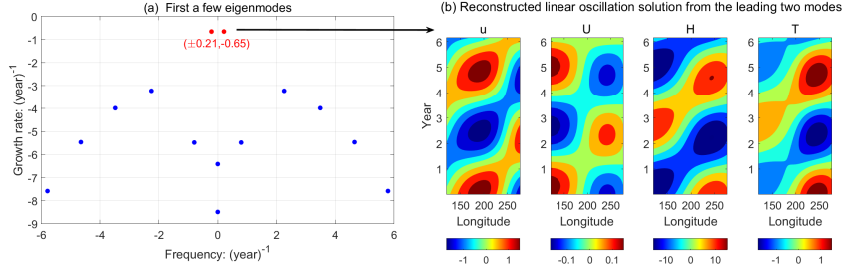}
\caption{Eigenvalue analysis of the spatially-extended stochastic system. Panel (a): the first a few eigenvalues. The eigenvalues are complex numbers. The real and imaginary parts are shown in the y- and x-axis, representing the growth rate and the frequency, respectively. These eigenvalues are ranked according to the growth rate. The first two eigenvalues, which are a pair of complex conjugates, are marked in red color. Panel (b): The reconstructed linear oscillation solution from the leading two eigenmodes, which is given by taking the summation of the two products of each of the eigenvalues and the associated eigenvectors. The growth rate is set to be zero here to display the oscillation patterns. }\label{LinearSolution}
\end{figure}

\subsection{Comparison of the projected solution with the full solution}
Before developing a reduced-order conceptual model based on the eigenvalue analysis, as shown above, it is essential to study the similarity and the gap between the full solution from the spatially-extended stochastic system \eqref{Atmosphere_model}--\eqref{ap_equation} and that from the low-dimensional representation consisting of only the leading two eigenmodes, namely the projected solution.

Figure \ref{Full_Projection} includes a comparison of the SST between these two solutions. The projected solution (Panel (b)) is qualitatively similar to the full solution (Panel (a)) with an almost identical pattern and a nearly negligible residual error (Panel (c)). This is not surprising since the decaying rates of the remaining eigenmodes are more significant than those of the leading two components. Thus, the time evolution of the signal is overall explained by the two leading modes. In addition to the spatiotemporal pattern, the two key statistics, namely the power spectrum and the PDF, of the averaged SST over the eastern Pacific $T_E$ are almost indistinguishable between the full and the projected solutions. In particular, the power spectrum peaks between 3 and 7 years, which is consistent with the observed frequency of the ENSO cycles. The projected solution also recovers the non-Gaussian PDF of $T_E$, which has almost the same skewness and kurtosis as the full solution. Non-Gaussianity is a crucial feature of the ENSO statistics, which indicates a non-symmetry between the El Ni\~no and the La Ni\~na. The fat tail on the positive side corresponds to the super El Ni\~nos, which remains in the projected solution.

\begin{figure}[ht]
\hspace*{-0cm}\includegraphics[width=1.0\textwidth]{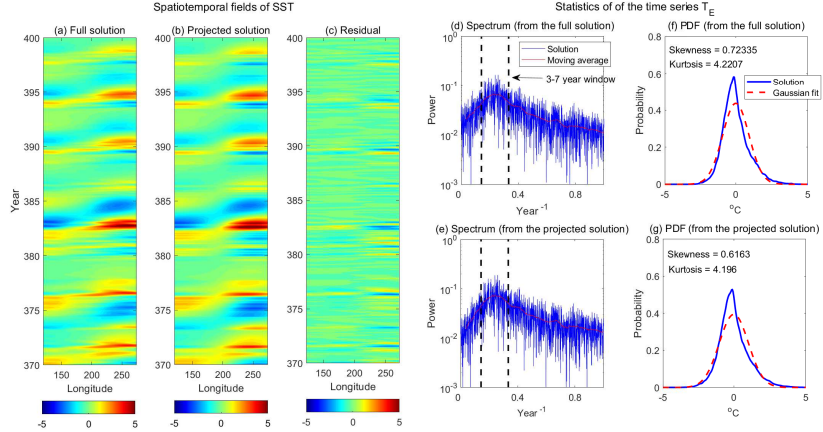}
\caption{Comparison of the full solution from the spatially-extended stochastic system \eqref{Atmosphere_model}--\eqref{ap_equation} with that from the low-order representation consisting of only the leading two eigenmodes, namely the projected solution. Panels (a)--(c): One realization of the full solution and the corresponding projected solution as well as their difference, namely the residual. Panels (d)--(g): the statistics of the time series $T_E$, which is the averaged SST over the eastern Pacific. Panels (d)--(e) compare the spectrum, where the two regions bounded by the two dashed black curves represent the 3- to the 7-year window. A running average showing the smoothed version of the spectrum is also included. Panels (f)--(g) compare the PDFs, where the red dashed curve shows the Gaussian fit. The values of the skewness and the kurtosis of $T_E$ are also included to indicate the strong non-Gaussian features. }\label{Full_Projection}
\end{figure}

Figure \ref{Case_DelayedEvent} includes a more detailed comparison between the full and the projected solutions by showing different physical and wave variables.  A delayed super El Ni\~no event (from $t=474$ to $t=476$) occurs during this period. Overall, the projected solution captures the spatiotemporal patterns of the full one. Especially, the delayed super El Ni\~no is reproduced in the projected solution. Yet, the small-scale wave propagation is missed in the projected solution. This is particularly unambiguous in the panels of the ocean Kelvin wave $K_O$ and the ocean Rossby wave $R_O$. As a result, the values across the western Pacific or the eastern Pacific remain constant for different fields. This is perhaps the most significant difference between the full and the projected solutions. Small-scale wave propagations are expected to appear at least with three degrees of freedom (e.g., either three variables or three regions across Pacific) \cite{chen2022multiscale}. Nevertheless, the projected solution clearly captures the large-scale eastward moving patterns, representing the ENSO development.

\begin{figure}[ht]
\hspace*{-0cm}\includegraphics[width=1.0\textwidth]{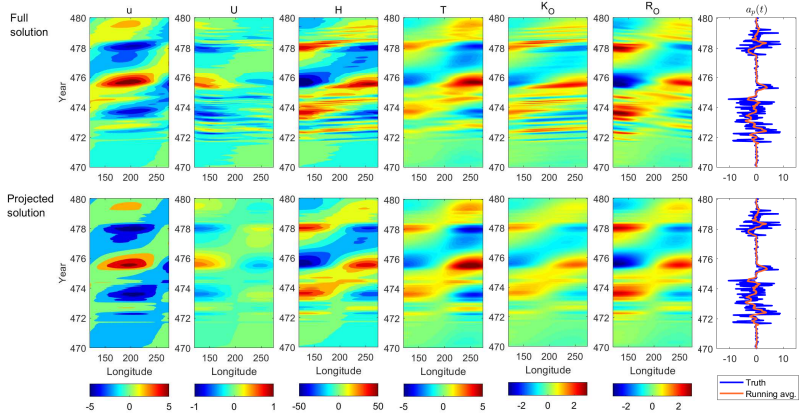}
\caption{A case study of the delayed super El Ni\~no event (from $t=474$ to $t=476$) for comparing the full solution from the spatially-extended stochastic system with that from the low-order representation consisting of only the leading two eigenmodes, namely the projected solution. The variables are the atmosphere zonal wind velocity $u$ (unit: m/s), the ocean zonal current $U$ (unit: m/s), the thermocline depth $H$ (unit: m), SST $T$ (unit: $^o$C), ocean Kelvin wave $K_O$ (no unit), ocean Rossby wave $R_O$ (no unit), the state of the Markov jump process, and the wind burst time series $a_p(t)$ (unit: m/s). In the panel of the wind burst amplitude $a_p(t)$, a 90-day running average of $a_p(t)$ (brown curve) is included to represent the lower frequency part of the wind bursts.}\label{Case_DelayedEvent}
\end{figure}

Figure \ref{Hovmoller} compares the full observed SST using the reanalysis data (Panel (a)) with the projected data using the leading two eigenmodes computed from the model. Panel (b) shows the results by directly projecting the SST to these two modes. The projected solution resembles the full observations in the eastern Pacific with a pattern correlation above $0.85$ at all longitude grid points. In contrast, the projected solution in the central-western Pacific differs from the full solution. This is not surprising since the starting model and the low-order presentation here focus only on the eastern Pacific El Ni\~no events and do not have the mechanism to create the central Pacific events. The projected solution in the western Pacific captures the overall patterns of the full observations, though the details are somewhat missed. In addition to using the SST data, the atmosphere wind data can also be used to reconstruct the low-order representation of the SST. See Panel (c). Here, the atmosphere wind data is first projected to the leading two eigenvectors corresponding to the atmosphere components to compute the projection coefficients. Then these coefficients are multiplied by the eigenvectors corresponding to the SST components for reconstruction. The reconstructed low-order representation of SST using the atmosphere reanalysis data has a similar pattern as that based on the direct projection using the SST in the eastern Pacific. Finally, Panel (d) shows the reconstructed SST data using a direct data-driven regression method. This is achieved by first computing the correlation coefficient $r(x)$ between the SST at each longitude with the Ni\~no 3 SST index using the observational data and then carrying out the spatiotemporal reconstruction using the formula $\mbox{SST}(x,t) = r(x)\mbox{Ni\~no 3}(t)$. Note that this procedure uses the entire observational data and can be regarded as one of the optimal linear reconstruction methods. It serves as a benchmark solution. The results show that, similar to the low-dimensional projected solutions in Panels (b)--(c), the reconstruction in the eastern Pacific resembles the truth while that in the central-western Pacific contains a significant error. The comparison between the results in Panels (b)--(c) and Panel (d) implies that the reconstruction using the model-induced basis functions is skillful in reproducing the eastern Pacific El Ni\~no events. It also indicates that both the wind and the SST data can be used for the reconstruction, allowing large freedom in choosing the available data sets.

\begin{figure}[ht]
\hspace*{-0cm}\includegraphics[width=1.0\textwidth]{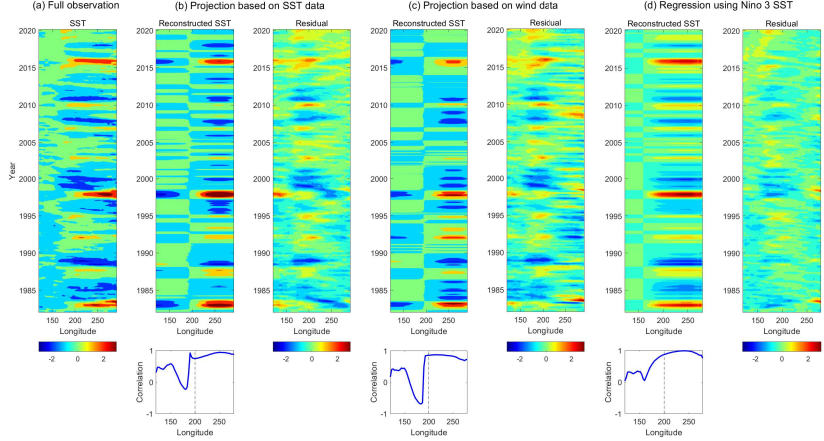}
\caption{Comparison of the observational SST data (Panel (a)) with low-order representations of the spatiotemporal solution using different reconstructed methods. Panel (b): the reconstruction by projecting the observed SST to the leading two eigenmodes from the spatially-extended system. Panel (c): the reconstruction by projecting the observed wind data to the leading two eigenmodes from the spatially-extended system and then recovering the SST patterns via the interconnection between the wind and the SST components in the eigenvectors. Panel (d): the reconstruction by using a linear regression to determine the SST. The regression coefficient is determined by computing the correlation coefficient between the SST at each longitude with the Ni\~no 3 SST index using the observational data. In Panels (b)--(d), the residual in the reconstruction, namely the difference between the truth shown in Panel (a) and the reconstructed spatiotemporal SST field, is also shown. In addition, the correlation between the reconstructed SST field with the truth at each longitude is included beneath the reconstructed field.  }\label{Hovmoller}
\end{figure}

\section{Reduced-Order Conceptual Models}\label{Sec:CM}
Given the fact that the reconstructed SST fields from the projected solution and the full solution are similar to each other, it is natural to develop a reduced-order conceptual model using these two eigenmodes to characterize ENSO complexity with low computational cost.

\subsection{The general framework}
Plugging the relationship of the matrix decomposition $\mathbf{M} = \mathbf{X}\mathbf{\Lambda}\mathbf{X}^{-1}$ in \eqref{MLambda} into the full system \eqref{LinearSystem_Stochastic} yields
\begin{equation}\label{LinearSystem_Stochastic3}
    \frac{\d\mathbf{u}}{\d t} = \mathbf{X}\mathbf{\Lambda}\mathbf{X}^{-1}\mathbf{u} + \mathbf{F},
\end{equation}
where $\mathbf{\Lambda}$ is a diagonal matrix.
Then multiplying \eqref{LinearSystem_Stochastic3} by $\mathbf{X}^{-1}$ gives $3N_O$ independent equations
\begin{equation}\label{LinearSystem_Stochastic4_temp2}
    \frac{\d\mathbf{v}}{\d t} =  \mathbf{\Lambda}\mathbf{v} + \mathbf{G},
\end{equation}
where 
\begin{equation}\label{LinearSystem_Stochastic4_temp1}
  \mathbf{v} = \mathbf{X}^{-1}\mathbf{u},\qquad\mbox{and}\qquad \mathbf{G} = \mathbf{X}^{-1}\mathbf{F}.
\end{equation}
Once $\mathbf{v}$ is solved, $\mathbf{u}$ can be obtained by
\begin{equation}\label{LinearSystem_Stochastic4_temp3}
  \mathbf{u} = \mathbf{X}\mathbf{v}.
\end{equation}

Denote by $\mathbf{z}^\mathtt{T}$ and $\bar{\mathbf{z}}^\mathtt{T}$ the first two rows of $\mathbf{X}^{-1}$. See Panel (b) of Figure \ref{z_and_su} for the profile of $\mathbf{z}$. Likewise, denote by $\mathbf{x}$ and $\bar{\mathbf{x}}$ the first two columns of $\mathbf{X}$. The corresponding leading two eigenvalues are denoted by $\lambda_o$ and $\bar\lambda_o$ (the two red dots in Panel (a) of Figure \ref{LinearSolution}), which are written as $-d_o \pm i\omega_o$ with $-d_o<0$ and $\omega_o$ being the growth rate and the frequency, respectively. Here $\cdot^\mathtt{T}$ is the transpose, and $\bar{\cdot}$ is the complex conjugate. Clearly, $\mathbf{z}^\mathtt{T}\mathbf{x} =\bar{\mathbf{z}}^\mathtt{T}\bar{\mathbf{x}} = 1$.
\begin{proposition}
Define $v = \mathbf{z}^\mathtt{T}\mathbf{u}$ and $g = \mathbf{z}^\mathtt{T}\mathbf{F}=\mathbf{z}^\mathtt{T}\mathbf{s}_\mathbf{u}a_p$
and naturally, $\bar{v} = \bar{\mathbf{z}}^\mathtt{T}\mathbf{u}$ and $\bar{g} = \bar{\mathbf{z}}^\mathtt{T}\mathbf{F}=\bar{\mathbf{z}}^\mathtt{T}\mathbf{s}_\mathbf{u}a_p$.
In light of the first two dimensions of \eqref{LinearSystem_Stochastic4_temp2}, the following two-dimensional forced oscillator system is reached,
\begin{equation}\label{LinearSystem_Stochastic6}
\begin{split}
    \frac{\d v}{\d t} &=  \lambda_ov + g,\\
    \frac{\d\bar{v}}{\d t} &=  \bar{\lambda}_o\bar{v} + \bar{g}.
\end{split}
\end{equation}
Once $v$ and $\bar{v}$ are computed, the approximate of $\mathbf{u}$ based on these two modes (namely the low-order representation of the spatiotemporal reconstructed solution), which is denoted by $\widetilde{\mathbf{u}}$, can be computed via
\begin{equation}\label{state_reconstruction_temp}
  \widetilde{\mathbf{u}} = v\mathbf{x} + \bar{v}\bar{\mathbf{x}}.
\end{equation}
\end{proposition}
Figure \ref{Eigenvector} shows the eigenvector associated with $\lambda_o$ for different model components.

It is worthwhile to highlight that, as was shown in \eqref{define_u}, the solution $\widetilde{\mathbf{u}}$ consists of the ocean Kelvin wave $K_O$, the ocean Rossby wave $R_O$, and the SST $T$. Recall in \eqref{physical_reconstruction} that the ocean current $U$ and the thermocline depth $H$ are the linear combinations of $K_O$ and $R_O$. In addition, the atmosphere wind $u$ and the potential temperature $\theta$ can be solved by exploiting the linear relationship with the ocean and SST variables according to \eqref{Combined_system} (see Appendix for details). Therefore, the two-dimensional low-order conceptual model \eqref{LinearSystem_Stochastic6} provides the solutions for all the atmosphere, ocean, and SST fields with detailed spatiotemporal solutions. This contrasts with many existing conceptual models, which focus only on a small set of specified state variables and often merely involve the averaged quantities over a large domain.

\begin{figure}[ht]
\hspace*{-0cm}\includegraphics[width=1.0\textwidth]{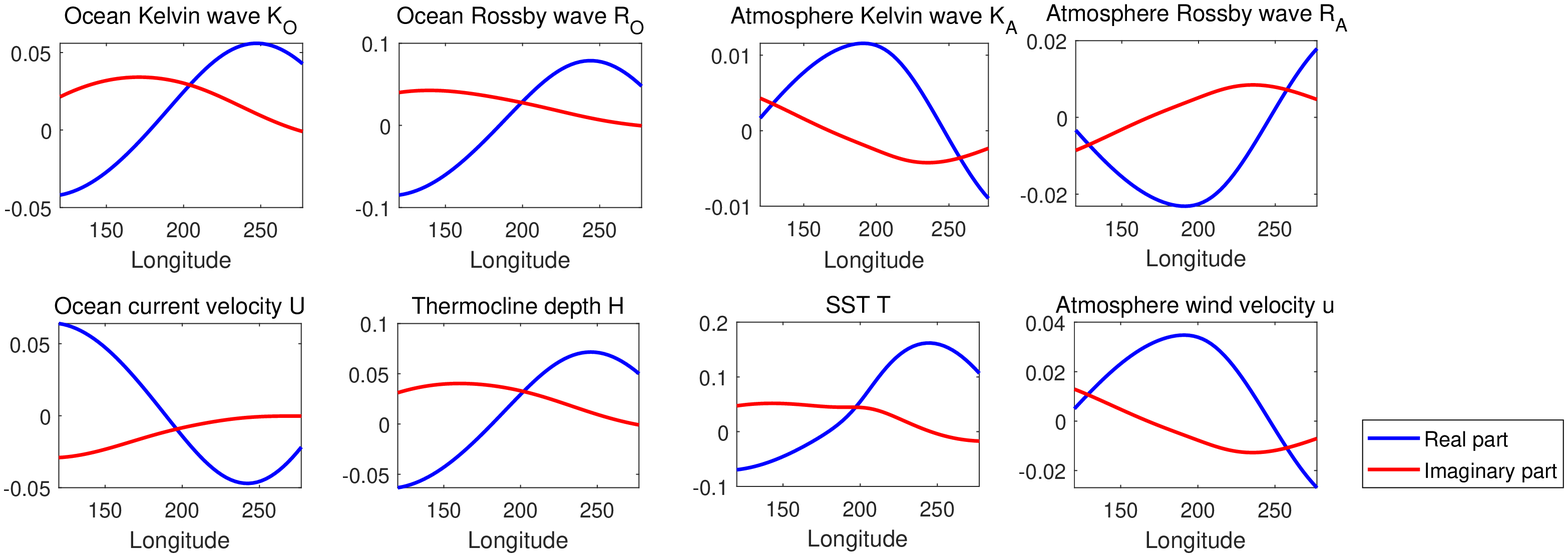}
\caption{The eigenvector associated with $\lambda_o$ corresponds to different variable components. }\label{Eigenvector}
\end{figure}

\subsection{Two-dimensional forced oscillator with arbitrarily specified quantities as state variables}
The general framework in \eqref{LinearSystem_Stochastic6} also allows the development of the reduced-order models with arbitrarily specified quantities (in addition to the direct model output $u$, $\theta$, $H$, $U$ and $T$) as state variables. This can be achieved by introducing two new real-valued state variables $\widetilde{a}$ and $\widetilde{b}$, which are the linear combinations of $v$ and $\bar{v}$ by angle rotations and amplifications,
\begin{subequations}\label{def_tilde}
\begin{align}
  \widetilde{a} &= \left(e^{i\alpha}v + e^{-i\alpha}\bar{v} \right)A_a,\label{def_tildea}\\
  \widetilde{b} &= \left(e^{i\beta}v + e^{-i\beta}\bar{v} \right)A_b.\label{def_tildeb}
\end{align}
\end{subequations}
Here $\alpha$ and $\beta$ are rotating angles to be in phase with the two new physical variables of interest while the real variables $A_a$ and $A_b$ are the amplitudes. Here, it is required that $\alpha,\beta\in[-\pi,\pi]$ and $\beta\neq\alpha$ to rule out the degenerated case.
The two variables $\widetilde{a}$ and $\widetilde{b}$ can be defined as any variables that can be represented by a linear combination of the model state variables.

\begin{proposition}\label{Prop:reconstruction_u_arbitrary}
The approximate solution of $\mathbf{u}$ based on the two arbitrary state variables $\widetilde{a}$ and $\widetilde{b}$  is given by
\begin{equation}\label{reconstruction_u_arbitrary_state_variables}
  {\widetilde{\mathbf{u}}} = \mathbf{x}_{\tilde{a}}\widetilde{a} + \mathbf{x}_{\tilde{b}}\widetilde{b},
\end{equation}
where
\begin{small}
\begin{equation}\label{Eigvec_new}
\begin{split}
\mathbf{x}_{\tilde{a}} &= \frac{1}{A_a}\left(\frac{e^{i\alpha}\mathbf{x}}{e^{i2\alpha}-e^{i2\beta}} + \frac{e^{-i\alpha}\bar{\mathbf{x}}}{e^{-i2\alpha}-e^{-i2\beta}}\right),\\
\mathbf{x}_{\tilde{b}} &= - \frac{1}{A_b}\left(\frac{e^{i\beta}\mathbf{x}}{e^{i2\alpha}-e^{i2\beta}} + \frac{e^{-i\beta}\bar{\mathbf{x}}}{e^{-i2\alpha}-e^{-i2\beta}}\right)
\end{split}
\end{equation}
\end{small}
are the new spatial bases (eigenvectors) associated with $\widetilde{a}$ and $\widetilde{b}$, which are real-valued.
\end{proposition}
\begin{proof}
Multiplying \eqref{def_tildea} by $A_be^{i\alpha}$ and multiplying \eqref{def_tildeb} by $A_ae^{i\beta}$ lead to
\begin{equation*}
\begin{split}
  A_be^{i\alpha}\widetilde{a} &= A_aA_be^{i2\alpha}v + A_aA_b \bar{v},\\
  A_ae^{i\beta}\widetilde{b} &= A_aA_be^{i2\beta}v + A_aA_b \bar{v}.
\end{split}
\end{equation*}
Taking the difference between the above two equations yields
\begin{equation}\label{expression_v}
  v = \frac{1}{A_aA_b}\frac{A_be^{i\alpha}\widetilde{a}-A_ae^{i\beta}\widetilde{b}}{e^{i2\alpha}-e^{i2\beta}},\qquad\mbox{and}\qquad \bar{v} = \frac{1}{A_aA_b}\frac{A_be^{-i\alpha}\widetilde{a}-A_ae^{-i\beta}\widetilde{b}}{e^{-i2\alpha}-e^{-i2\beta}}.
\end{equation}
Recall in \eqref{state_reconstruction_temp} that the full spatiotemporal evolution of the solution is given by $\widetilde{\mathbf{u}} = v\mathbf{x} + \bar{v}\bar{\mathbf{x}}$. Plugging \eqref{expression_v} into \eqref{state_reconstruction_temp} gives
\begin{equation}\label{state_reconstruction3}
\begin{split}
  {\widetilde{\mathbf{u}}} &= { v\mathbf{x} + \bar{v}\bar{\mathbf{x}}}\\
  &= \frac{1}{A_aA_b}\frac{A_be^{i\alpha}\widetilde{a}-A_ae^{i\beta}\widetilde{b}}{e^{i2\alpha}-e^{i2\beta}}\mathbf{x} + \frac{1}{A_aA_b}\frac{A_be^{-i\alpha}\widetilde{a}-A_ae^{-i\beta}\widetilde{b}}{e^{-i2\alpha}-e^{-i2\beta}}\bar{\mathbf{x}}\\
  & = \frac{1}{A_a}\left(\frac{e^{i\alpha}\mathbf{x}}{e^{i2\alpha}-e^{i2\beta}} + \frac{e^{-i\alpha}\bar{\mathbf{x}}}{e^{-i2\alpha}-e^{-i2\beta}}\right)\widetilde{a} \\
  &\quad\ - \frac{1}{A_b}\left(\frac{e^{i\beta}\mathbf{x}}{e^{i2\alpha}-e^{i2\beta}} + \frac{e^{-i\beta}\bar{\mathbf{x}}}{e^{-i2\alpha}-e^{-i2\beta}}\right)\widetilde{b}\\
   { :}&{ = \mathbf{x}_{\tilde{a}}\widetilde{a} + \mathbf{x}_{\tilde{b}}\widetilde{b}}.
\end{split}
\end{equation}
\end{proof}

The above framework is adapted to the cases where the main interest lies in certain averaged quantities, as in many of the classical conceptual models for the ENSO. For example, $\widetilde{a}$ can be the SST in the eastern Pacific and $\widetilde{b}$ the thermocline in the western Pacific. They can also be the atmosphere wind and ocean current averaged over the entire Pacific or a specific region. In such situations, the amplitudes $A_a$ and $A_b$ can be determined by two scalars $x_{\tilde{a}}$ and $x_{\tilde{b}}$, which are the averaged values of $\mathbf{x}_{\tilde{a}}$ and $\mathbf{x}_{\tilde{b}}$ over a certain domain. This allows the development of a two-dimensional system for $\widetilde{a}$ and $\widetilde{b}$ that is linked with the traditional conceptual models.

\begin{proposition}\label{Prop:TwoDimensionalModel}
Let $x_a := A_ae^{i\alpha}$ and $x_b := A_be^{i\beta}$ such that in \eqref{def_tilde} $\widetilde{a} = x_av + \bar{x}_a\bar{v}$ and $\widetilde{b} = x_bv + \bar{x}_b\bar{v}$. Then the two-dimensional conceptual model for $\widetilde{a}$ and $\widetilde{b}$ yields
\begin{equation}\label{Two_dimensional_model}
\begin{split}
  \frac{\d\widetilde{a}}{\d t}  &= -d_o  \widetilde{a} + \widetilde{c}_{11}\widetilde{a} + \widetilde{c}_{12}\widetilde{b} + 2g_{a,Re},\\
  \frac{\d\widetilde{b}}{\d t}  &= -d_o  \widetilde{b} + \widetilde{c}_{21}\widetilde{a} + \widetilde{c}_{22}\widetilde{b} + 2g_{b,Re},
\end{split}
\end{equation}
where all the coefficients are real. In \eqref{Two_dimensional_model},
\begin{equation}\label{C_Matrix_Det}
  \widetilde{C} =\left(
    \begin{array}{cc}
      \widetilde{c}_{11} & \widetilde{c}_{12} \\
      \widetilde{c}_{21} & \widetilde{c}_{22} \\
    \end{array}
  \right)=
  -\frac{\omega_o}{1-c_{Re} d_{Re}}\left(
     \begin{array}{cc}
       c_{Re}d_{Im} & c_{Im} \\
       d_{Im} & d_{Re}c_{Im} \\
     \end{array}
   \right).
\end{equation}
with $c = c_{Re} + c_{Im} := \frac{x_a}{x_b}$ and $d = d_{Re} + d_{Im} := \frac{x_b}{x_a}$. The two forcing terms are given by $g_{a,Re} = Re(gx_a) = Re(\mathbf{z}^T\mathbf{s}_\mathbf{u}x_a)a_p$ and $g_{b,Re} = Re(gx_b)=Re(\mathbf{z}^T\mathbf{s}_\mathbf{u}x_b)a_p$.
\end{proposition}
See the Appendix for the proofs.

\subsection{Special case: stochastic discharge-recharge oscillator}
The discharge-recharge oscillator argues that discharge and recharge of equatorial heat content cause the coupled system to oscillate \cite{jin1997equatorial}. Define $\widetilde{a}:=T_E$ and $\tilde{b}:=H_W$ in \eqref{Two_dimensional_model}, representing the averaged SST in the eastern Pacific and the averaged thermocline depth in the western Pacific. The corresponding scalars $x_{T_E}$ and $x_{H_W}$ are given by
\begin{equation}\label{scalars_discharge}
    x_{T_E}=0.1302+0.0106i \qquad \mbox{and} \qquad x_{H_W}=-0.0279+0.0375i,
\end{equation}
which are computed by averaging the eigenvectors associated with $T$ and $H$ over the eastern Pacific and western Pacific, respectively. See Figure \ref{basis}.

\begin{figure}[ht]
\centering\hspace*{-0cm}\includegraphics[width=0.5\textwidth]{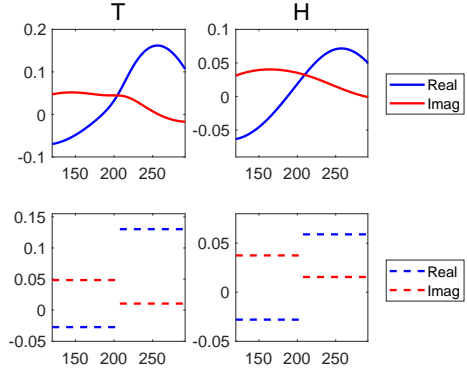}
\caption{The eigenvectors of $v$ associated with $T$ and $H$ (top) and their averaged values over the eastern and western Pacific (bottom). }\label{basis}
\end{figure}

The other coefficients in \eqref{Two_dimensional_model} can be computed using the eigenvalues of the leading two modes $\lambda_o$ as well as $x_{T_E}$ and $x_{H_W}$ in \eqref{scalars_discharge}. They are:
\begin{equation}\label{Coefficient_C}
\widetilde{C}=\left(\begin{array}{cc}
\widetilde{c}_{11} & \widetilde{c}_{12} \\
\widetilde{c}_{21} & \widetilde{c}_{22}
\end{array}\right)=\left(\begin{array}{cc}
0.0752 & 0.3965 \\
-0.0509 & -0.0752
\end{array}\right)
\end{equation}
and
\begin{equation}\label{Coefficient_alpha}
\alpha_{T_E}=2 \operatorname{Re}\left(\mathbf{z}^T \mathbf{s}_{\mathbf{u}} x_{T_E}\right)=1.0094, \quad \alpha_{H_W}=2 \operatorname{Re}\left(\mathbf{z}^T \mathbf{s}_{\mathbf{u}} x_{H_W}\right)=-0.4217.
\end{equation}
The remaining thing is to determine the multiplicative noise coefficient $\sigma_p(T_W)$ of the stochastic amplitude $a_p$ in \eqref{ap_equation}, which is a function of the averaged SST in the western Pacific $T_W$. Yet, $T_W$ is not the state variable in the two-dimensional discharge-recharge oscillator model. Nevertheless, $T_W$ can be written as a linear function of $T_E$ and $H_W$.
Recall that
\begin{equation}
T_E=x_{T_E} v+\bar{x}_{T_E} \bar{v} \qquad\mbox{and}\qquad H_W =x_{H_W} v+\bar{x}_{H_W} \bar{v},
\end{equation}
which can be used to solve $\mbox{Re}(v)$ and $\mbox{Im}(v)$,
\begin{equation}
\begin{split}
\mbox{Re}(v)&=\frac{T_E \mbox{Im}(x_{H_W})-H_W \mbox{Im}(x_{T_E})}{2(\mbox{Re}(x_{T_E})\mbox{Im}(x_{H_W}) - \mbox{Im}(x_{T_E}) \mbox{Re}(x_{H_W}))}, \\
\mbox{Im}(v)&=\frac{T_E \mbox{Re}(x_{H_W})-H_W \mbox{Re}(x_{T_E})}{2(\mbox{Re}(x_{T_E})\mbox{Im}(x_{H_W}) - \mbox{Im}(x_{T_E}) \mbox{Re}(x_{H_W}))}. \\
\end{split}
\end{equation}
Once the basis $v$ is determined, the value of $T_W$ can be computed,
\begin{equation}
T_W=x_{T_W} v+\bar{x}_{T_W} \bar{v}
\end{equation}
which are functions of $H_W$ and $T_E$, and $x_{T_W}$ is given by
\begin{equation}
x_{T_W} = -0.0272 + 0.0484i.
\end{equation}
The direct representation of the $T_W$ by $T_E$ also justifies and unifies the use of SST in different regions (averaged values over the eastern Pacific, over the western Pacific or basin averaged value) as multiplicative noise coefficients in various conceptual models \cite{an2020fokker, an2020enso, levine2017simple, giorgini2022non}.

After obtaining the value of $T_W$, the multiplicative noise coefficient 
$\sigma_p(T_W)$ is determined. Note that the $\sigma_p(T_W)$ is slightly modified to $\sigma_{p 0}=0.2$ and $\sigma_{p 1}=1.8$ to compensate for the approximation error from the low-order representation. 

According to \eqref{Coefficient_C}, the coefficient $\widetilde{c}_{11}$ and $\widetilde{c}_{12}$ associated with $T_E$ equations are both positive. The coefficient $\widetilde{c}_{11}>0$ represents the Bjerknes positive feedback while $\widetilde{c}_{12}>0$ is the recharge mechanism. On the other hand, a negative coefficient $\widetilde{c}_{21}$ appears in the equation of $H_W$, which indicates the discharge mechanism. Another finding is that, although the wind bursts directly impact the SST in heuristic thinking, the stochastic forcing appears in both the equations of $T_E$ and $H_W$ for the completeness of the stochastic discharge-recharge diagram. The opposite signs of $\alpha_{T_E}$ and $\alpha_{H_W}$ imply that the increase of SST in the eastern Pacific by the WWBs will simultaneously cause a decrease of the thermocline in the western Pacific.

Figure \ref{time_series_final} shows the time series and phase plot of both the deterministic and stochastic versions of the discharge-recharge models developed above. The deterministic model (by setting the damping to be zero and ignoring the stochastic terms) leads to a regular oscillator while the time series of $H_W$ and $T_E$ generated from the stochastic model mimic those in observations.
It is worth highlighting that the phase difference between $H_W$ and $T_E$ is $122^o$ (not $\pi/2$ or $90^o$), which is also confirmed by the lag between the two time series. Such a phase difference also resembles that from the observations. In addition, the observational time series and the stochastic model simulation of $T_E$ are compared. Due to the random noise, the trajectories from the stochastic model and the observations do not expect to have a one-to-one point-wise match between each other. Nevertheless, the stochastic model simulation succeeds in capturing the EP events. Finally, the PDF and the autocorrelation function (ACF) of $T_E$ from the stochastic model have very similar behavior to the observations. The PDF is skewed with a one-sided fat tail, which indicates that the model can reproduce a sufficient number of the super El Ni\~no events. Similarly, matching the ACF with observations implies that the time series of $T_E$ from the model has, on average, the same temporal correlation (namely the time scale of the memory effect) as the observations. These findings suggest that the low-order model can capture the observed dynamical and statistical features of the time series $T_E$, which is an important indicator for the eastern Pacific El Ni\~no.

\begin{figure}[ht]
\centering\hspace*{-0cm}\includegraphics[width=1\textwidth]{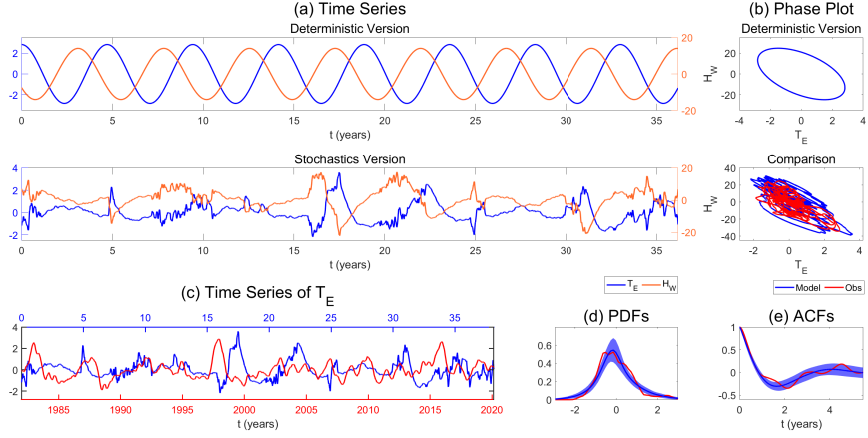}
\caption{ Comparison of the time series and the phase plot between the deterministic and the stochastic versions of the discharge-recharge model; the time series, PDFs, and ACFs of $T_E$ between the stochastic version and the observations. Panel(a): The time series of $T_E$ and $H_W$ generated from two versions of the discharge-recharge model. Panel (b): The phase plot of $T_E$ and $H_W$ from two versions of the model, and observations. Here, the $H_W$ in two version models are multiplied by a constant to keep the same standard deviation of the observations. Panel (c): The time series of $T_E$ from the stochastic version of the model and observations. Panel (d): The PDFs of $T_E$ for the solution to the stochastic version of the model and observations. Panel (e): The ACFs of $T_E$ from the stochastic version of the model and observations. In Panel (d)-(e), red and blue curves are for the observation and stochastic model, respectively. For the stochastic model, the total 20000-year-long simulation is divided into 500 non-overlapping segments, each of which has a 40-year period as the observation. Then the average (blue line) and its one standard deviation intervals (shading) are illustrated.}\label{time_series_final}
\end{figure}

\subsection{Special case: stochastic delayed oscillator}
The two-dimensional stochastic oscillator \eqref{Two_dimensional_model} can be further reduced to a one-dimensional oscillator with a delayed effect, which leads to the delayed oscillator of the ENSO.
\begin{proposition}\label{Prop:delayed_oscillator}
Defining $\widetilde{a}:=T_E$, the averaged SST over the eastern Pacific. The delayed oscillator is given by
\begin{equation}\label{Delayed_General}
\begin{split}
  \frac{\d T_E(t)}{\d t} & =-d_o  T_E(t) + \widetilde{c}_{11}T_E(t) +\widetilde{c}_{12}c_ye^{-d_o\varphi/\omega_o}T_E(t-\varphi/\omega_o) + \alpha_{T_E}a_p(t)\\
    &\qquad+ 2c_y\widetilde{c}_{12}\int_{t-\varphi/\omega_o}^t e^{-d_o(t-s)}\bigg[y_{1,Re}\Big(\cos(\omega)\alpha_{T_E,1}-\sin(\omega)\alpha_{T_E,2}\Big)\\
   &\qquad\qquad\qquad-y_{1,Im}\Big(\cos(\omega)\alpha_{T_E,2}+\sin(\omega)\alpha_{T_E,1}\Big)\bigg]a_p(s)\d s
\end{split}
\end{equation}
where $\varphi:=\alpha-\beta$ and $\omega = \omega_o(t-s)-\varphi$. The coefficients $\widetilde{c}_{11}$, $\widetilde{c}_{12}$ and $\alpha_{T_E}$ are the same as those in \eqref{Coefficient_C} and \eqref{Coefficient_alpha}.
The other two noise coefficients $\alpha_{T_E,1}$ and $\alpha_{T_E,2}$ are given by
\begin{equation*}
\begin{split}
  \alpha_{T_E,1} &= Re\Big(r_{1}2Re(\mathbf{z}^T\mathbf{s}_\mathbf{u}x_{T_E}) + r_{2}2Re(\mathbf{z}^T\mathbf{s}_\mathbf{u}x_{H_W}) \Big),\\
  \alpha_{T_E,2} &= Im\Big(r_{1}2Re(\mathbf{z}^T\mathbf{s}_\mathbf{u}x_{T_E}) + r_{2}2Re(\mathbf{z}^T\mathbf{s}_\mathbf{u}x_{H_W}) \Big),
\end{split}
\end{equation*}
Here, $r_1$ and $r_2$ are solved as follows. As the matrix $\widetilde{C}$ in \eqref{C_Matrix_Det} has two distinct eigenvalues, which are also complex conjugates, it can be diagonalized $\widetilde{C} = \mathbf{Y}\Lambda\mathbf{Y}^{-1}$. Denote by
\begin{equation*}
 \mathbf{Y} = \left(
                \begin{array}{cc}
                 y_1  & \bar{y}_1 \\
                 y_2  & \bar{y}_2 \\
                \end{array}
              \right)\qquad\mbox{and}\qquad
 \mathbf{Y}^{-1} = \left(
                \begin{array}{cc}
                  r_1 & r_2 \\
                  \bar{r}_1 & \bar{r}_2 \\
                \end{array}
              \right)
\end{equation*}
where the fact that the two eigenvectors are complex conjugates has been utilized. In addition, the angle difference $\varphi$ and the amplitude difference $c_y$ are given by
\begin{equation}\label{angle_amplitude}
  M_{y_2} = c_yM_{y_1}\varPhi,
\end{equation}
where
\begin{equation*}
\begin{split}
  M_{y_1} = \left(
    \begin{array}{cc}
      y_{1,Re} & -y_{1,Im} \\
      -y_{1,Im} & -y_{1,Re} \\
    \end{array}
  \right)\qquad&
  M_{y_2} = \left(
    \begin{array}{cc}
      y_{2,Re} & -y_{2,Im} \\
      -y_{2,Im} & -y_{2,Re} \\
    \end{array}
  \right),\\
  \varPhi = \left(
    \begin{array}{cc}
      \cos\varphi & \sin\varphi \\
      -\sin\varphi & \cos\varphi \\
    \end{array}
  \right),\qquad& c_y = \frac{\|y_2\|}{\|y_1\|},
\end{split}
\end{equation*}
with $y_1 = y_{1,Re} + iy_{1,Im}$ and $y_2 = y_{2,Re} + iy_{2,Im}$.
\end{proposition}
See the Appendix for the proofs. Numerically, these constants and coefficients are given by
\begin{equation*}
\begin{gathered}
  y_1 = 0.9414,\  y_2 = -0.1785 - i0.2861,\   r_1 = 0.5311+i0.3314,\  r_2 = i1.7476,\\
  \alpha_{T_E} = 1.0094,\ 
  \alpha_{T_E,1} = 0.5361,\ \alpha_{T_E,2} = 0.4024,\ 
  \varphi = 2.1287, \  c_y = 0.3582.
\end{gathered}
\end{equation*}

The detailed expression of the stochastic forcing with the memory effect derived here is an essential supplement to the original delayed oscillator model \cite{suarez1988delayed, battisti1989interannual}, which was deterministic. Several conclusions can be drawn from this delayed oscillator. First, In front of the delayed term $T_E\left(t-\varphi/\omega_o\right)$, there is a factor with an exponential decay $e^{-d_o\varphi/\omega_o}$. This is a crucial feature, which indicates that the contribution of the delayed signal will be affected by the damping effect and becomes weaker as time goes on. Second, there are two ways that stochastic noise contributes to the system. The term $\alpha_{T_E}a_p(t)$ provides an instantaneous response to the system, while the stochastic forcing inside the integral represents the contribution from the noise accumulation. The latter also has a delayed contribution. Notably, although the damping time of the noise $d_p$ is intraseasonal, the noise interacting with the coupled system can affect the interannual time scale. Third, similar to the discharge-recharge oscillator, the delay here is $\varphi=2.1287$ ($122^o$) instead of $\pi/2$ ($90^o$). Finally, the sign of the delayed term depends on the sign of $\widetilde{c}_{12}$, which can be either positive or negative, depending on the state variables used for the model. The parameter $\widetilde{c}_{12}$ is positive in the above model since $\varphi=2.1287>\pi/2$. Yet, when the angle $\varphi$ is less than $\pi/2$, then $\widetilde{c}_{12}<0$ representing a negative delay, which is consistent with the discussions in \cite{battisti1989interannual}. In particular, if the average of the thermocline over the whole Pacific is adopted as was in \cite{battisti1989interannual}, then a negative delay (about $\pi/6$) appears.

\section{Conclusion and Discussions}\label{Sec:Conclusion}
In this paper, a framework of low-order stochastic conceptual models for the ENSO is systematically derived from a spatially-extended stochastic dynamical system with full mathematical rigor. It provides outputs for all the atmosphere, ocean, and SST components with detailed spatiotemporal patterns, which differs from many existing conceptual models focusing only on a small set of specified state variables. The framework also allows rigorously deriving the stochastic versions of the discharge-recharge oscillator and delayed oscillator models. Observational data is utilized to validate the skill of the low-order conceptual model in capturing the large-scale features of the ENSO in the eastern Pacific. The rigorous derivation of these low-order models provides a unique connection between models with different complexities. The framework developed here also facilitates understanding the instantaneous and memory effects of stochastic noise in contributing to the large-scale dynamics of ENSO.

The current modeling framework focuses on the eastern Pacific ENSO events. To characterize the ENSO diversity \cite{capotondi2015understanding} that also includes the central Pacific El Ni\~no, the ocean zonal advection mechanism should be incorporated in the starting spatially-extended model. It is worth noting that the dynamics of the central Pacific El Ni\~no often contain certain levels of nonlinearity \cite{zhao2021breakdown, chen2022multiscale}. Therefore, suitable nonlinearity needs to be added to the linear low-order conceptual models to better describe the central Pacific El Ni\~no, which is remained as a future work.

Despite the simple forms of conceptual models, they can be utilized as testbeds for guiding more sophisticated models. In particular, including stochastic noise or stochastic parameterizations in ICMs or GCMs is always a challenging topic. Therefore, testing various strategies based on conceptual models is practically useful. These simple conceptual models can also be utilized to build a hierarchy of models with different complexity for understanding the ENSO dynamics with the help of data-driven techniques.

\section*{Acknowledgements}
The research of N.C. is partially funded by ONR N00014-21-1-2904. Y.Z. is supported as a PhD research assistant under this grant. N.C. wish to thank his late advisor Andrew J. Majda for valuable discussions on some initial ideas of this work.

\section{Appendix}
\subsection{A summary of the parameters in the spatially-extended stochastic system}

This section will provide details for parameter values and dimensional values.
\begin{table}[H]
\centering
\scriptsize
\begin{tabular}{|ccc|}
\hline
Variable & Explanation & Value\\
\hline
$H$ & thermocline depth anomalies & $20.8 \mathrm{~m}$ \\
$T$ & sea surface temperature anomalies & $1.5 \mathrm{~K}$ \\
$U$ & zonal current speed anomalies & $0.25 \mathrm{~m} / \mathrm{s}$ \\
$t$ & time axis intraseasonal & 0.48 days \\
$u$ & zonal wind speed anomalies & $5 \mathrm{~m} / \mathrm{s}$\\
$\theta$ & potential temperature anomalies & $1.5 \mathrm{~K}$\\
\hline
\end{tabular}\caption{Physical dimensions of model parameters.}\label{Table:variable}
\end{table}

\begin{table}[H]
\centering
\scriptsize
 \begin{tabular}{|ccc|}
 \hline
   Parameter & Explanation & Value\\
   \hline
   %
   $\chi_A$ & atmospheric meridional projection coefficient & 0.325 \\
   $\chi_O$ & oceanic meridional projection coefficient & 1.3\\
   $L_A$ & equatorial belt length & 8/3 \\
   $L_O$ & equatorial Pacific length & 7/6 \\
   $\bar{Q}$ & mean vertical moisture gradient & 0.9 \\
   $\bar{T}$ & mean SST & 16.67 (which is 25$^o$C)\\
   $\gamma$ & wind stress coefficient & 6.53 \\
   $\zeta$ & latent heating exchange coefficient & 8.7 \\
   $\alpha_q$ & latent heating factor & $\alpha_q=q_cq_e\exp(q_e\bar{T})/\tau_q$ \\
   $q_c$ & latent heating multiplier coefficient & 7 \\
   $q_e$ & latent heating exponential coefficient & 0.093 \\
   $\tau_q$ & latent heating adjustment rate & 15 \\
   $r_W$ & western boundary reflection coefficient & 0.5 \\
   $r_E$ & eastern boundary reflection coefficient & 1 \\
   $c_1$ & modified ratio of phase speed & 0.5 \\
   $\eta$ & profile of thermocline feedback & $\eta(x)=1.5+\tanh(7.5(x-L_O/2))/2$ \\
   $d_A$ & additional damping & $10^{-8}$\\
   $d_p$ & wind burst damping & 3.4 (which is 3mon$^{-1}$) \\
   $s_p$ & wind burst zonal structure & $s_p(x)=\exp(-45(x-L_O/4)^2)$ \\
   $\sigma_{p0}$ & wind burst source of quiescent state & 0.2 \\
   $\sigma_{p1}$ & wind burst source of active state & 2.6 \\
   $\mu_{10}$ & transition rate from the active to quiescent state & $\mu_{10}= \left(1-\tanh \left(2 T_W\right)\right)/4$\\
   $\mu_{01}$ & transition rate from the quiescent to active state & $\mu_{01}= \left(\tanh \left(2 T_W\right)+1\right)/8$\\
 \hline
 \end{tabular} \caption{Model parameter values.}\label{Table:param}
 \end{table}

\subsection{Determining the matrix $\mathbf{M}$ in \eqref{LinearSystem}}
Recall the system in \eqref{LinearSystem}. Rewrite it in terms of the components with respect to the three set of the variables $\mathbf{K_O}$, $\mathbf{R_O}$ and $\mathbf{T}$:
\begin{equation}\label{LinearSystemMatrixForm}
    \frac{\d}{\d t}\left(
                  \begin{array}{c}
                    \mathbf{K_O} \\
                    \mathbf{R_O} \\
                    \mathbf{T} \\
                  \end{array}
                \right) = \left(
                            \begin{array}{ccc}
                              \mathbf{M}_{11} & \mathbf{M}_{12} & \mathbf{M}_{13} \\
                              \mathbf{M}_{21} & \mathbf{M}_{22} & \mathbf{M}_{23} \\
                              \mathbf{M}_{31} & \mathbf{M}_{32} & \mathbf{M}_{33} \\
                            \end{array}
                          \right) \left(
                  \begin{array}{c}
                    \mathbf{K_O} \\
                    \mathbf{R_O} \\
                    \mathbf{T} \\
                  \end{array}
                \right),
\end{equation}
where the matrix $\mathbf{M}$ in \eqref{LinearSystem} is expanded into a $3\times 3$ matrix. In the following, the goal is to determine the matrix $\mathbf{M}$, namely the $\mathbf{M}_{ij}$ with $i,j = 1, 2, 3$, in \eqref{LinearSystemMatrixForm}.

Let us start with the ocean model
\begin{equation}\label{Ocean_appendix}
\begin{aligned}
    \partial_{t}K_O + c_1\partial_xK_O &= \frac{a}{2}\ (K_A - R_A),\qquad\qquad &&K_O(0,t) = r_WR_O(0,t),\\
    \partial_{t}R_O - \frac{c_1}{3}\partial_xR_O &= -\frac{a}{3}\ (K_A - R_A),\qquad\qquad&&R_O(L_O,t) = r_EK_O(L_O,t),\\
\end{aligned}
\end{equation}
An {upwind scheme} is utilized for the spatial discretization,
\begin{equation}\label{Ocean_upwind_appendix}
\begin{split}
    \partial_t K_{O,i} &= -\frac{c_1}{\Delta x}\left(K_{O,i}-K_{O,{i-1}}\right) + \frac{a}{2}\left(K_{A,i}-R_{A,i}\right),\qquad i=1,\ldots, N_O,\\
    \partial_t R_{O,i} &= \frac{c_1}{3\Delta x}\left(R_{O,{i+1}}-R_{O,i}\right) - \frac{a}{3}\left(K_{A,i}-R_{A,i}\right),\qquad i=1,\ldots, N_O,
\end{split}
\end{equation}
where the boundary conditions are given by $K_{O,0} = r_W R_{O,1}$ and $R_{O,{N_O+1}} = r_E K_{O,{N_O}}$. Therefore,
\begin{tiny}
\begin{equation}\label{matrix_1_4_appendix}{
\begin{gathered}
    \mathbf{M}_{11} = - \frac{c_1}{\Delta x}
    \left(
      \begin{array}{cccccc}
        1   &       &       &       &       &   \\
        -1  & 1     &       &       &       &   \\
            & -1    &  1    &       &       &   \\
            &       & \ddots& \ddots&       &   \\
            &       &       & -1    & 1     &   \\
            &       &       &       &   -1  & 1 \\
      \end{array}
    \right)_{N_O\times N_O},\ 
    \mathbf{M}_{12} = - \frac{c_1}{\Delta x}
    \left(
      \begin{array}{cccccc}
        -r_W    &       &       &       &       &   \\
                &       &       &       &       &   \\
                &       &       &       &       &   \\
                &       &       &       &       &   \\
                &       &       &       &       &   \\
                &       &       &       &       &   \\
      \end{array}
    \right)_{N_O\times N_O}\\
    \mathbf{M}_{22} =  \frac{c_1}{3\Delta x}
    \left(
      \begin{array}{cccccc}
        -1  & 1      &        &         &       &   \\
            & -1     & 1      &         &       &   \\
            &        & \ddots & \ddots  &       &   \\
            &        &        & -1      & 1     &   \\
            &        &        &         & -1    & 1 \\
            &        &        &         &       & -1 \\
      \end{array}
    \right)_{N_O\times N_O},\quad
    \mathbf{M}_{21} =  \frac{c_1}{3\Delta x}
    \left(
      \begin{array}{cccccc}
                &       &       &       &       &   \\
                &       &       &       &       &   \\
                &       &       &       &       &   \\
                &       &       &       &       &   \\
                &       &       &       &       &   \\
                &       &       &       &       & r_E  \\
      \end{array}
    \right)_{N_O\times N_O}
\end{gathered}}
\end{equation}    
\end{tiny}

Next, for the atmosphere model
\begin{equation}\label{Atmosphere_appendix}
\begin{aligned}
    d_AK_A+\partial_xK_A & = m_1 (\mathbf{1}_{[0,L_O]}\alpha_qT-\langle E_q\rangle)/\alpha_q,\qquad\qquad&&K_A(0,t) = K_A(L_A,t),\\
    d_AR_A-\partial_xR_A/3 & = m_2 (\mathbf{1}_{[0,L_O]}\alpha_qT-\langle E_q\rangle)/\alpha_q,\qquad\qquad&& R_A(0,t) = R_A(L_A,t),
\end{aligned}
\end{equation}
where $\mathbf{1}_{[0,L_O]}$ is an indicator function with value being $1$ when variables are located within the interval $[0,L_O]$ and $0$ otherwise.
The discretization of \eqref{Atmosphere_appendix} results in
\begin{equation*}{
\begin{split}
    d_A K_{A,i} + \frac{1}{\Delta x}(K_{A,{i+1}} - K_{A,{i}}) &= m_1\left(T_i - \frac{1}{N_A}\sum_{j=1}^{N_O}T_j\right),\\
    d_A R_{A,i} - \frac{1}{3\Delta x}(R_{A,{i+1}} - R_{A,{i}}) &= m_2\left(T_i - \frac{1}{N_A}\sum_{j=1}^{N_O}T_j\right),\qquad i = 1,\ldots, N_O,
\end{split}}
\end{equation*}
and
\begin{equation}\label{Atmosphere_discrete_appendix}{
\begin{split}
    d_A K_{A,i} + \frac{1}{\Delta x}(K_{A,{i+1}} - K_{A,{i}}) &= m_1\left(- \frac{1}{N_A}\sum_{j=1}^{N_O}T_j\right),\\
    d_A R_{A,i} - \frac{1}{3\Delta x}(R_{A,{i+1}} - R_{A,{i}}) &= m_2\left(- \frac{1}{N_A}\sum_{j=1}^{N_O}T_j\right),\qquad i = N_O+1,\ldots, N_A,
\end{split}}
\end{equation}
where the relationship $E_q=\alpha_q T$ has been used. Note that the averaging is taken over the entire equator band $[0, N_A]$ and therefore the factor $1/N_A$ in front of the summation appears.
Rearranging the terms in \eqref{Atmosphere_discrete_appendix} gives
\begin{equation*}{
\begin{split}
    K_{A,{i+1}} + (d_A\Delta x - 1) K_{A,{i}}  &= \frac{N_A-1}{N_A}\Delta x m_1\left(T_i - \frac{1}{N_A-1}\sum_{ 1\leq j\leq N_O}^{j\neq i}T_j\right),\\
    - R_{A,{i+1}} + (3d_A\Delta x + 1) R_{A,i} &= 3\frac{N_A-1}{N_A}\Delta x m_2\left(T_i - \frac{1}{N_A-1}\sum_{ 1\leq j\leq N_O}^{j\neq i}T_j\right), \\
    &\qquad\qquad\qquad\qquad\qquad\qquad\qquad i = 1,\ldots N_O
\end{split}}
\end{equation*}
and
\begin{equation*}{
\begin{split}
    K_{A,{i+1}} + (d_A\Delta x - 1) K_{A,{i}} &= \frac{N_A-1}{N_A}\Delta x m_1\left( - \frac{1}{N_A-1}\sum_{ 1\leq j\leq N_O}T_j\right),\\
    - R_{A,{i+1}} + (3d_A\Delta x + 1) R_{A,i} &= 3\frac{N_A-1}{N_A} \Delta x m_2\left( - \frac{1}{N_A-1}\sum_{ 1\leq j\leq N_O}T_j\right),\\
    &\qquad\qquad\qquad\qquad\qquad\  i = N_O+1,\ldots, N_A
\end{split}}
\end{equation*}
Thus, $K_{A,i}$ and $R_{A,i}$ for $i = 1,\ldots N_O$ can be solved via the following linear systems,
\begin{equation}\label{M_T_appendix}
\begin{split}
    \mathbf{M_K}\cdot\mathbf{K_A} &= \widetilde{C}_1\mathbf{B_K}\mathbf{T},\\
    \mathbf{M_R}\cdot\mathbf{R_A} &= \widetilde{C}_2\mathbf{B_R}\mathbf{T},
\end{split}
\end{equation}
which means $\mathbf{K_A}$ and $\mathbf{R_A}$ can be expressed by $\mathbf{T}$. In \eqref{M_T_appendix}
\begin{equation*}
    \widetilde{C}_1=\frac{N_A-1}{N_A}\Delta x m_1\qquad\mbox{and}\qquad \widetilde{C}_2=3\frac{N_A-1}{N_A}\Delta x m_2
\end{equation*}
with $\Delta{x}$ being one unit of the spatial discretization.
Further define $\kappa=d_A\Delta x$. Then the matrices in \eqref{M_T_appendix} are given by
\begin{tiny}
\begin{equation*}{
    \mathbf{M_K} = \left(
      \begin{array}{cccccccc}
        -1+\kappa&  1        &       &           &           &           &  \\
                & -1+\kappa  &   1   &           &           &           &  \\
                &           &\ddots &   \ddots  &           &           &   \\
                &           &       & -1+\kappa  &   1       &           &  \\
                &           &       &           &\ddots     &  \ddots   &  \\
                &           &       &           &           &-1+ \kappa   &  1 \\
        1      &           &       &           &           &           & -1+ \kappa \\
      \end{array}
    \right)_{N_A\times N_A},\ 
    \mathbf{K_A} = \left(
      \begin{array}{c}
        K_{A,1} \\
        K_{A,2} \\
        \vdots\\
        K_{A,{N_O}} \\
        \vdots \\
        K_{A,{N_A-1}}\\
        K_{A,{N_A}} \\
      \end{array}
    \right)_{N_A\times 1}},
\end{equation*}
\begin{equation*}{
    \mathbf{M_R} = \left(
      \begin{array}{cccccccc}
        1+3\kappa& -1        &       &           &           &           &  \\
                & 1+3\kappa  &  -1   &           &           &           &  \\
                &           &\ddots &   \ddots  &           &           &   \\
                &           &       & 1+3\kappa  &  -1       &           &  \\
                &           &       &           &\ddots     &  \ddots   &  \\
                &           &       &           &           &1+3\kappa   & -1 \\
        -1      &           &       &           &           &           & 1+3\kappa \\
      \end{array}
    \right)_{N_A\times N_A},\ 
    \mathbf{R_A} = \left(
      \begin{array}{c}
        R_{A,1} \\
        R_{A,2} \\
        \vdots\\
        R_{A,{N_O}} \\
        \vdots \\
        R_{A,{N_A-1}}\\
        R_{A,{N_A}} \\
      \end{array}
    \right)_{N_A\times 1}}.
\end{equation*}       
\end{tiny}
and
\begin{equation*}{
    \mathbf{B}_{\mathbf{K}} = \mathbf{B_R} = \left(
                              \begin{array}{ccccccc}
                                1                   & -\frac{1}{N_A-1}  & \ldots & -\frac{1}{N_A-1}  \\
                                -\frac{1}{N_A-1}    & 1                 & \ldots & -\frac{1}{N_A-1}  \\
                                \vdots              & \vdots            & \ddots & \vdots            \\
                                -\frac{1}{N_A-1}    & -\frac{1}{N_A-1}  & \ldots & 1                 \\
                                -\frac{1}{N_A-1}    & -\frac{1}{N_A-1}  & \ldots & -\frac{1}{N_A-1}  \\
                                \vdots              & \vdots            & \ddots & \vdots            \\
                                -\frac{1}{N_A-1}    & -\frac{1}{N_A-1}  & \ldots & -\frac{1}{N_A-1}  \\
                              \end{array}
                            \right)_{N_A\times N_O},\ 
    \mathbf{T}=\left(
      \begin{array}{c}
        T_1 \\
        \vdots \\
        T_{N_O} \\
      \end{array}
    \right)_{N_O\times 1}}.
\end{equation*}
Recall the ocean model,
\begin{equation*}
\begin{split}
    \partial_t K_{O,i} &= -\frac{c_1}{\Delta x}\left(K_{O,i}-K_{O,{i-1}}\right) { + \frac{a}{2}\left(K_{A,i}-R_{A,i}\right)},\\
    \partial_t R_{O,i} &= \frac{c_1}{3\Delta x}\left(R_{O,{i+1}}-R_{O,i}\right) { - \frac{a}{3}\left(K_{A,i}-R_{A,i}\right)},
\end{split}
\end{equation*}
Since $\mathbf{K_A}$ and $\mathbf{R_A}$ can be expressed by $\mathbf{T}$, we have
\begin{equation}\label{matrix_13_23_appendix}
\begin{split}
    \mathbf{M}_{13} &= \frac{a}{2}\left.\left[\widetilde{C}_1\mathbf{M_K}^{-1}\mathbf{B_K} - \widetilde{C}_2\mathbf{M_R}^{-1}\mathbf{B_R}\right]\right|_{\mbox{Row~}\{1:N_O\}},\\
    \mathbf{M}_{23} &= -\frac{a}{3}\left.\left[\widetilde{C}_1\mathbf{M_K}^{-1}\mathbf{B_K} - \widetilde{C}_2\mathbf{M_R}^{-1}\mathbf{B_R}\right]\right|_{\mbox{Row~}\{1:N_O\}},
\end{split}
\end{equation}
where the original matrices on the right hand side should be of size $N_A\times N_O$ but only the first $N_O$ rows are utilized to form $\mathbf{M}_{13}$ and $\mathbf{M}_{23}$ which correspond to $K_{A,1},\ldots,K_{A,N_O}$ and $R_{A,1},\ldots,R_{A,N_O}$ within the Pacific band.

Finally, let's consider the SST model,
\begin{equation*}
\partial_{t}T =-bT + c_1\eta (K_O + R_O),
\end{equation*}
the discrete form of which is straightforward,
\begin{equation*}
    \partial_{t}T_i=-bT_i + c_1\eta (K_{O,i} + R_{O,i}).
\end{equation*}
The matrices $\mathbf{M}_{31}$, $\mathbf{M}_{32}$ and $\mathbf{M}_{33}$ are all diagonal and their $(i,i)$-th diagonal entries are given by
\begin{equation}\label{matrix_31_32_appendix}
    (\mathbf{M}_{31})_{ii} = (\mathbf{M}_{32})_{ii} = c_1\eta(x_i),\qquad (\mathbf{M}_{33})_{ii} = -b.
\end{equation}
Collecting the results in \eqref{matrix_1_4_appendix}, \eqref{matrix_13_23_appendix} and \eqref{matrix_31_32_appendix} yields the matrix $\mathbf{M}$.

\subsection{Proofs of the propositions}
\subsubsection{Proof of Proposition \ref{Prop:TwoDimensionalModel}}
\begin{proof}
Let $\widetilde{a} = a + \bar{a}$ and $\widetilde{b} = b + \bar{b}$, where $a = vx_a$, $\bar{a} = \bar{v}\bar{x}_a$, $b=vx_b$, and $\bar{b}=\bar{v}\bar{x}_b$. It is clear that $a$ and $b$ are complex scalars. For convenience, write
\begin{equation}\label{a_b_real_imag}
  a = a_{Re} + i a_{Im}, \qquad \mbox{and}\qquad b = b_{Re} + i b_{Im},
\end{equation}
and therefore
\begin{equation*}
\begin{split}
  \bar{a}  = a_{Re} - i a_{Im},&\qquad\mbox{and}\qquad \bar{b} =   b_{Re} - i b_{Im}\\
  \widetilde{a} = a+\bar{a}=2a_{Re},&\qquad\mbox{and}\qquad \widetilde{b} = b+\bar{b}=2b_{Re}.
\end{split}
\end{equation*}
Once $a_{Re}$ is computed, multiplying it by a factor of $2$ finishes the solutions $\widetilde{a}$ and $\widetilde{b}$. For this reason, the focus here is on solving $a_{Re}$.

Multiplying \eqref{LinearSystem_Stochastic6} by ${x}_a$ and ${x}_b$, respectively, yields
\begin{subequations}\label{evolution_a_b}
\begin{align}
  \frac{\d a}{\d t} &= (-d_o + i\omega_o) a + g_a,\label{evolution_a_b1}\\
  \frac{\d b}{\d t} &= (-d_o + i\omega_o) b + g_b,\label{evolution_a_b2}
\end{align}
\end{subequations}
where $g_a = gx_a = g_{a,Re} + ig_{a,Im}$ and $g_b = gx_b= g_{b,Re} + ig_{b,Im}$. Inserting \eqref{a_b_real_imag} into \eqref{evolution_a_b} results in four equations for the real and imaginary parts of $a$ and the real and imaginary parts of $b$,
\begin{subequations}
\begin{align}
  \frac{\d a_{Re}}{\d t} &= -d_o a_{Re} - \omega_o a_{Im} + g_{a,Re},\label{Evolution_a_real}\\
  \frac{\d a_{Im}}{\d t} &= -d_o a_{Im} + \omega_o a_{Re} + g_{a,Im},\label{Evolution_a_imag}\\
  \frac{\d b_{Re}}{\d t} &= -d_o b_{Re} - \omega_o b_{Im} + g_{b,Re},\label{Evolution_b_real}\\
  \frac{\d b_{Im}}{\d t} &= -d_o b_{Im} + \omega_o b_{Re} + g_{b,Im}.\label{Evolution_b_imag}
\end{align}
\end{subequations}
Here, only the two equations \eqref{Evolution_a_real} and \eqref{Evolution_b_real} for the real part of $a$ and $b$ need to be solved, since $2a_{Re}$ and $2b_{Re}$ give the pair of physical variables we aim at.

The goal now is to use $a_{Re}$ and $b_{Re}$ to represent the term $a_{Im}$ and $b_{Im}$ in \eqref{Evolution_a_real} and \eqref{Evolution_b_real}. To this end, let's introduce two new complex scalars
\begin{equation}\label{Define_c_d}
  c = c_{Re} + c_{Im} := \frac{x_a}{x_b}, \qquad\mbox{and}\qquad d = d_{Re} + d_{Im} := \frac{x_b}{x_a},
\end{equation}
which are the ratio (and the inverse ratio) of the two bases $x_a$ and $x_b$. Combining \eqref{a_b_real_imag} and \eqref{Define_c_d} yields,
\begin{equation*}
\begin{split}
  a_{Re} + ia_{Im} &= (b_{Re} + ib_{Im})(c_{Re} + ic_{Im})=(b_{Re}c_{Re}-b_{Im}c_{Im})+i(b_{Re}c_{Im}+b_{Im}c_{Re}),\\
  b_{Re} + ib_{Im} &= (a_{Re} + ia_{Im})(d_{Re} + id_{Im})=(a_{Re}d_{Re}-a_{Im}d_{Im})+i(a_{Re}d_{Im}+a_{Im}d_{Re}),
\end{split}
\end{equation*}
which leads to
\begin{subequations}\label{a_b_Im}
\begin{align}
   a_{Im} &=  c_{Im}b_{Re} + c_{Re}b_{Im},\label{a_Im}\\
   b_{Im} &=  d_{Im}a_{Re} + d_{Re}a_{Im}.\label{b_Im}
\end{align}
\end{subequations}
With \eqref{a_b_Im} in hand, plugging \eqref{a_Im} into \eqref{Evolution_a_real} yields
\begin{equation}\label{Evolution_a_real_2}
  \frac{\d a_{Re}}{\d t} = -d_o a_{Re} - \omega_o c_{Im} b_{Re} - \omega_o c_{Re}b_{Im} + g_{a,Re}.
\end{equation}
Next, plugging \eqref{b_Im} into \eqref{Evolution_a_real_2} to cancel $b_{Im}$ leads to
\begin{equation}\label{Evolution_a_real_3}
  \frac{\d a_{Re}}{\d t} = -d_o a_{Re} - \omega_o c_{Im}b_{Re} - \omega_o c_{Re}d_{Im}a_{Re} - \omega_o c_{Re} d_{Re}a_{Im} + g_{a,Re}.
\end{equation}
Although \eqref{Evolution_a_real_3} again introduces $a_{Im}$, making use of \eqref{Evolution_a_real} and \eqref{Evolution_a_real_3} suffices to cancel this $a_{Im}$. To this end, multiplying the equation \eqref{Evolution_a_real} by $c_{Re} d_{Re}$ gives
\begin{equation}\label{Evolution_a_real_4}
  c_{Re} d_{Re}\frac{\d a_{Re}}{\d t}  = -d_o c_{Re} d_{Re}a_{Re} - \omega_o c_{Re} d_{Re}a_{Im} + c_{Re} d_{Re}g_{a,Re}.
\end{equation}
Subtracting \eqref{Evolution_a_real_4} from \eqref{Evolution_a_real_3} results in
\begin{equation*}
  (1-c_{Re} d_{Re})\frac{\d a_{Re}}{\d t}  = -d_o (1-c_{Re} d_{Re}) a_{Re} - \omega_o c_{Im}b_{Re} - \omega_o c_{Re}d_{Im}a_{Re} + (1-c_{Re} d_{Re})g_{a,Re},
\end{equation*}
which implies
\begin{equation}\label{Evolution_a_real_5}
  \frac{\d a_{Re}}{\d t}  = -d_o  a_{Re} - \frac{\omega_o c_{Im}}{1-c_{Re} d_{Re}}b_{Re} - \frac{\omega_o c_{Re}d_{Im}}{1-c_{Re} d_{Re}}a_{Re} + g_{a,Re}.
\end{equation}

Applying the same argument leads to the corresponding equation of $b_{Re}$ that does not explicitly depend on $b_{Im}$,
\begin{equation}\label{Evolution_b_real_5}
  \frac{\d b_{Re}}{\d t}  = -d_o  b_{Re} - \frac{\omega_o d_{Im}}{1-d_{Re} c_{Re}}a_{Re} - \frac{\omega_o d_{Re}c_{Im}}{1-d_{Re} c_{Re}}b_{Re} + g_{b,Re}.
\end{equation}

Finally, multiplying both $a_{Re}$ and $b_{Re}$ by a factor of $2$ gives the two-dimensional oscillator model of $\tilde{a}$ and $\tilde{b}$,
\begin{equation}\label{Discharge_Recharge_General}
\begin{split}
  \frac{\d\widetilde{a}}{\d t}  &= -d_o  \widetilde{a} - \frac{\omega_o c_{Im}}{1-c_{Re} d_{Re}}\widetilde{b} - \frac{\omega_o c_{Re}d_{Im}}{1-c_{Re} d_{Re}}\widetilde{a} + 2g_{a,Re},\\
  \frac{\d\widetilde{b}}{\d t}  &= -d_o  \widetilde{b} - \frac{\omega_o d_{Im}}{1-d_{Re} c_{Re}}\widetilde{a} - \frac{\omega_o d_{Re}c_{Im}}{1-d_{Re} c_{Re}}\widetilde{b} + 2g_{b,Re},
\end{split}
\end{equation}
where
\begin{equation}\label{g_discharge_recharge}
  g_{a,Re} = Re(gx_a) = Re(\mathbf{z}^T\mathbf{s}_\mathbf{u}x_a)a_p\qquad\mbox{and}\qquad g_{b,Re} = Re(gx_b)=Re(\mathbf{z}^T\mathbf{s}_\mathbf{u}x_b)a_p.
\end{equation}
It is convenient to define
\begin{equation}
\begin{split}
  \widetilde{C} &=\left(
    \begin{array}{cc}
      \widetilde{c}_{11} & \widetilde{c}_{12} \\
      \widetilde{c}_{21} & \widetilde{c}_{22} \\
    \end{array}
  \right)=
  -\frac{\omega_o}{1-c_{Re} d_{Re}}\left(
     \begin{array}{cc}
       c_{Re}d_{Im} & c_{Im} \\
       d_{Im} & d_{Re}c_{Im} \\
     \end{array}
   \right) \\
   &=-\frac{\omega_o}{1-\cos^2(\alpha-\beta)}\left(
     \begin{array}{cc}
       -\sin(\alpha-\beta)\cos(\alpha-\beta) & \frac{A_a}{A_b}\sin(\alpha-\beta) \\
       -\frac{A_b}{A_a}\sin(\alpha-\beta) & \sin(\alpha-\beta)\cos(\alpha-\beta) \\
     \end{array}
   \right),
\end{split}
\end{equation}
which links $(c, d)$ with $(\alpha, \beta)$ and gives the form in \eqref{Two_dimensional_model}.
\end{proof}
\subsubsection{Proof of Proposition \ref{Prop:delayed_oscillator}}
\begin{proof}
The delayed oscillator is formed based on the same starting idea from the discharge-recharge oscillator. It exploits the equations in \eqref{Discharge_Recharge_General} to represent $\tilde{b}(\tau)$ in the first equation as a function of $\tilde{a}(\tau-\delta)$ leads to a delayed El Ni\~no.

Let us start with \eqref{Discharge_Recharge_General}. To make the notation simple, rewrite \eqref{Discharge_Recharge_General} as follows:
\begin{equation}\label{Discharge_Recharge_General2}
\begin{split}
  \frac{\d\widetilde{a}}{\d t}  &= -d_o  \widetilde{a} + \widetilde{c}_{11}\widetilde{a} +\widetilde{c}_{12}\widetilde{b} + \widetilde{g}_{a},\\
  \frac{\d\widetilde{b}}{\d t}  &= -d_o  \widetilde{b} + \widetilde{c}_{21}\widetilde{a} +\widetilde{c}_{22}\widetilde{b} + \widetilde{g}_{b},
\end{split}
\end{equation}
where $\widetilde{g}_{a}=2g_{a,Re}$ and $\widetilde{g}_{b}=2g_{b,Re}$. Define two-dimensional vectors $\widetilde{u}, \widetilde{g}$ and a $2\times2$ matrix $ {\mathbf{C}}$,
\begin{equation*}
   {\mathbf{u}} = \left(
                    \begin{array}{c}
                      \widetilde{a} \\
                      \widetilde{b} \\
                    \end{array}
                  \right),\qquad
   {\mathbf{g}} = \left(
                    \begin{array}{c}
                      \widetilde{g}_a \\
                      \widetilde{g}_b \\
                    \end{array}
                  \right),\qquad
   {\mathbf{C}} = -d_o\mathbf{I}+\left(
                    \begin{array}{cc}
                      \widetilde{c}_{11} & \widetilde{c}_{12}\\
                      \widetilde{c}_{21} & \widetilde{c}_{22}\\
                    \end{array}
                  \right)=-d_o\mathbf{I}+\widetilde{C}.
\end{equation*}
It is easy to check that $ {\mathbf{C}}$ has two distinct eigenvectors $-d_o \pm i\omega_o$. Thus, a similarity transform leads to
${\mathbf{C}} = \mathbf{Y} {\mathbf{\Lambda}}\mathbf{Y}^{-1}$, where $ {\mathbf{\Lambda}}$ is a diagonal matrix with diagonal entries being $ -d_o\pm i\omega_o$. Now, \eqref{Discharge_Recharge_General2} can be written as
\begin{equation}\label{Delayed_Eqn1}
  \frac{\d{\mathbf{u}}}{\d t} = {\mathbf{C}}{\mathbf{u}} + {\mathbf{g}}
  = \mathbf{Y}{\mathbf{\Lambda}}\mathbf{Y}^{-1}{\mathbf{u}} + {\mathbf{g}}.
\end{equation}
Define ${\mathbf{w}}=\mathbf{Y}^{-1}{\mathbf{u}}$ and ${\mathbf{f}} = \mathbf{Y}^{-1}{\mathbf{g}}$. Multiplying \eqref{Delayed_Eqn1} by $\mathbf{Y}^{-1}$ yields
\begin{equation*}
  \frac{\d{\mathbf{w}}}{\d t} = {\mathbf{\Lambda}}{\mathbf{w}} + {\mathbf{f}},
\end{equation*}
where
\begin{equation*}
  {\mathbf{w}} = \left(
                    \begin{array}{c}
                      w \\
                      \bar{w} \\
                    \end{array}
                  \right),\qquad
  {\mathbf{f}} = \left(
                    \begin{array}{c}
                      f \\
                      \bar{f} \\
                    \end{array}
                  \right),\qquad
  {\mathbf{\Lambda}} = \left(
                    \begin{array}{cc}
                      -d_o+i\omega_o&0 \\
                      0&-d_o-i\omega_o \\
                    \end{array}
                  \right),
\end{equation*}
and therefore
\begin{equation}\label{Equation_v}
  \frac{\d w}{\d t} = (-d_o+i\omega_o)w + f.
\end{equation}
Once $\mathbf{w}$ is solved, $\mathbf{u}$ is easily recovered by $\mathbf{u} = \mathbf{Y}\mathbf{w}$.
Since the two eigenvalues are a complex conjugate pair, the eigenvector matrix $\mathbf{Y}$ must have the following form
\begin{equation}\label{Matrix_Y}
  \mathbf{Y} = \left(
                 \begin{array}{cc}
                   y_1 & \bar{y}_1 \\
                   y_2 & \bar{y}_2 \\
                 \end{array}
               \right).
\end{equation}
Therefore,
\begin{equation}\label{Define_ab}
  \widetilde{a} := u_1 = 2Re(y_1w),\qquad\mbox{and}\qquad \widetilde{b} := u_2 = 2Re(y_2w).
\end{equation}
For the notation simplicity below, rewrite $y_1$ and $y_2$ into the following form
\begin{equation*}
  y_1 = y_{1,Re} + iy_{1,Im},\qquad\mbox{and}\qquad y_2 = y_{2,Re} + iy_{2,Im}.
\end{equation*}
Similarly, rewrite the initial value $w(0)$ and the forcing $f(s)$ as
\begin{equation*}
  w(0) = w_{Re}(0) + iw_{Im}(0),\qquad\mbox{and}\qquad f(s) = f_{Re}(0) + i f_{Im}(0).
\end{equation*}
Thus, the solution of \eqref{Equation_v} can be written as
\begin{equation*}
\begin{split}
  w(t) &= w(0) e^{-d_ot } \Big(\cos(\omega_o t) + i\sin(\omega_ot)\Big) \\
  &\qquad\qquad + \int_0^t e^{-d_o(t-s)}\Big(\cos\big(\omega_o(t-s)\big)+i\sin\big(\omega_o(t-s)\big)\Big)f(s)ds\\
  &=e^{-d_ot} \Big[\Big(w_{Re}(0) \cos(\omega_ot) - w_{Im}(0) \sin(\omega_ot)\Big) \\
  &\qquad\qquad\qquad\qquad\qquad+i \Big(w_{Re}(0) \sin(\omega_ot) + w_{Im}(0) \cos(\omega_ot)\Big) \Big]\\
  &\qquad +\int_0^t e^{-d_o(t-s)}\Big[\Big(\cos\big(\omega_o(t-s)\big)f_{Re}(s)-\sin\big(\omega_o(t-s)\big)f_{Im}\Big) \\
   &\qquad\qquad+i \Big(\cos\big(\omega_o(t-s)\big)f_{Im}(s)+\sin\big(\omega_o(t-s)\big)f_{Re}\Big)\Big]ds.
\end{split}
\end{equation*}
Therefore,
\begin{equation}\label{u1_form}
\begin{split}
  u_1(t) &= 2Re(y_1w)\\
   &= 2e^{-d_ot}\Big[y_{1,Re}\Big(w_{Re}(0)\cos(\omega_ot)-w_{Im}(0)\sin(\omega_ot)\Big) \\
   & \qquad\qquad\qquad\qquad\qquad\qquad- y_{1,Im}\Big(w_{Re}(0)\sin(\omega_ot)+w_{Im}(0)\cos(\omega_ot)\Big)\Big]\\
   &\qquad+2\int_0^t e^{-d_o(t-s)}\Big[y_{1,Re}\Big(\cos\big(\omega_o(t-s)\big)f_{Re}(s)-\sin\big(\omega_o(t-s)\big)f_{Im}(s)\Big)\\
   &\qquad\qquad-y_{1,Im}\Big(\cos\big(\omega_o(t-s)\big)f_{Im}(s)+\sin\big(\omega_o(t-s)\big)f_{Re}(s)\Big)
   \Big]ds.
\end{split}
\end{equation}
Similarly,
\begin{equation}\label{u2_form}
\begin{split}
  u_2(t) &= 2Re(y_2w)\\
   &= 2e^{-d_ot}\Big[y_{2,Re}\Big(w_{Re}(0)\cos(\omega_ot)-w_{Im}(0)\sin(\omega_ot)\Big) \\
   & \qquad\qquad\qquad\qquad\qquad\qquad- y_{2,Im}\Big(w_{Re}(0)\sin(\omega_ot)+w_{Im}(0)\cos(\omega_ot)\Big)\Big]\\
   &\qquad+2\int_0^t e^{-d_o(t-s)}\Big[y_{2,Re}\Big(\cos\big(\omega_o(t-s)\big)f_{Re}(s)-\sin\big(\omega_o(t-s)\big)f_{Im}(s)\Big)\\
   &\qquad\qquad-y_{2,Im}\Big(\cos\big(\omega_o(t-s)\big)f_{Im}(s)+\sin\big(\omega_o(t-s)\big)f_{Re}(s)\Big)
   \Big]ds.
\end{split}
\end{equation}
Note that \eqref{u1_form} and \eqref{u2_form} have the same composition: a contribution from initial value and a contribution from external forcing. Below, let's write down these two components (and drop the common prefactors) for the convenience of discussion.

\noindent I. Composition of initial value:
\begin{equation*}
\begin{split}
  u_1^I = & \Big(y_{1,Re}w_{Re}(0) - y_{1,Im}w_{Im}(0)\Big)\cos(\omega_ot) - \Big(y_{1,Re}w_{Im}(0) - y_{1,Im}w_{Re}(0)\Big)\sin(\omega_ot)\\
  u_2^I = & \Big(y_{2,Re}w_{Re}(0) - y_{2,Im}w_{Im}(0)\Big)\cos(\omega_ot) - \Big(y_{2,Re}w_{Im}(0) - y_{2,Im}w_{Re}(0)\Big)\sin(\omega_ot)
\end{split}
\end{equation*}
II. Composition of external forcing:
\begin{equation*}
\begin{split}
  u_1^{II}=& \Big(y_{1,Re}f_{Re}(s) - y_{1,Im}f_{Im}(s)\Big)\cos(\omega_o(t-s)) - \Big(y_{1,Re}f_{Im}(s) - y_{1,Im}f_{Re}(s)\Big)\sin(\omega_o(t-s))\\
  u_2^{II}=& \Big(y_{2,Re}f_{Re}(s) - y_{2,Im}f_{Im}(s)\Big)\cos(\omega_o(t-s)) - \Big(y_{2,Re}f_{Im}(s) - y_{2,Im}f_{Re}(s)\Big)\sin(\omega_o(t-s))
\end{split}
\end{equation*}
For the initial value part, rewrite the two expressions as
\begin{subequations}
\begin{align}
  u_1^I&=\left(
           \begin{array}{cc}
             y_{1,Re}w_{Re}(0) - y_{1,Im}w_{Im}(0) & -y_{1,Re}w_{Im}(0) - y_{1,Im}w_{Re}(0) \\
           \end{array}
         \right)\left(
         \begin{array}{c}
         \cos(\omega_ot) \\
         \sin(\omega_ot) \\
         \end{array}
         \right)\notag\\
  &=\left(
           \begin{array}{cc}
             w_{Re}(0) & w_{Im}(0) \\
           \end{array}
         \right)
  \left(
    \begin{array}{cc}
      y_{1,Re} & -y_{1,Im} \\
      -y_{1,Im} & -y_{1,Re} \\
    \end{array}
  \right)\left(
         \begin{array}{c}
         \cos(\omega_ot) \\
         \sin(\omega_ot) \\
         \end{array}
         \right)\label{eqn_u_1_I}\\
  u_2^I&=\left(
           \begin{array}{cc}
             y_{2,Re}w_{Re}(0) - y_{2,Im}w_{Im}(0) & -y_{2,Re}w_{Im}(0) - y_{2,Im}w_{Re}(0) \\
           \end{array}
         \right)\left(
         \begin{array}{c}
         \cos(\omega_ot) \\
         \sin(\omega_ot) \\
         \end{array}
         \right)\notag\\
  &=\left(
           \begin{array}{cc}
             w_{Re}(0) & w_{Im}(0) \\
           \end{array}
         \right)
  \left(
    \begin{array}{cc}
      y_{2,Re} & -y_{2,Im} \\
      -y_{2,Im} & -y_{2,Re} \\
    \end{array}
  \right)\left(
         \begin{array}{c}
         \cos(\omega_ot) \\
         \sin(\omega_ot) \\
         \end{array}
         \right)\label{eqn_u_2_I}
\end{align}
\end{subequations}
and
\begin{subequations}
\begin{align}
  u_1^{II}&=\left(
           \begin{array}{cc}
             y_{1,Re}f_{Re}(s) - y_{1,Im}f_{Im}(s) & -y_{1,Re}f_{Im}(s) - y_{1,Im}f_{Re}(s) \\
           \end{array}
         \right)\left(
         \begin{array}{c}
         \cos(\omega_o(t-s)) \\
         \sin(\omega_o(t-s)) \\
         \end{array}
         \right)\notag\\
  &=\left(
           \begin{array}{cc}
             f_{Re}(s) & f_{Im}(s) \\
           \end{array}
         \right)
  \left(
    \begin{array}{cc}
      y_{1,Re} & -y_{1,Im} \\
      -y_{1,Im} & -y_{1,Re} \\
    \end{array}
  \right)\left(
         \begin{array}{c}
         \cos(\omega_o(t-s)) \\
         \sin(\omega_o(t-s)) \\
         \end{array}
         \right)\label{eqn_u_1_II}\\
  u_2^{II}&=\left(
           \begin{array}{cc}
             y_{2,Re}f_{Re}(s) - y_{2,Im}f_{Im}(s) & -y_{2,Re}f_{Im}(s) - y_{2,Im}f_{Re}(s) \\
           \end{array}
         \right)\left(
         \begin{array}{c}
         \cos(\omega_o(t-s)) \\
         \sin(\omega_o(t-s)) \\
         \end{array}
         \right)\notag\\
  &=\left(
           \begin{array}{cc}
             f_{Re}(s) & f_{Im}(s) \\
           \end{array}
         \right)
  \left(
    \begin{array}{cc}
      y_{2,Re} & -y_{2,Im} \\
      -y_{2,Im} & -y_{2,Re} \\
    \end{array}
  \right)\left(
         \begin{array}{c}
         \cos(\omega_o(t-s)) \\
         \sin(\omega_o(t-s)) \\
         \end{array}
         \right)\label{eqn_u_2_II}
\end{align}
\end{subequations}
It is clear that the difference between $u_1$ and $u_2$ is given by the angle $\varphi$ between the two matrices consisted of the eigenvectors $y_1, y_2$
\begin{equation}\label{M12_Matrix}
  M_{y_1} = \left(
    \begin{array}{cc}
      y_{1,Re} & -y_{1,Im} \\
      -y_{1,Im} & -y_{1,Re} \\
    \end{array}
  \right)\qquad\mbox{and}\qquad
  M_{y_2} = \left(
    \begin{array}{cc}
      y_{2,Re} & -y_{2,Im} \\
      -y_{2,Im} & -y_{2,Re} \\
    \end{array}
  \right),
\end{equation}
that is, there exists a rotation matrix and a constant
\begin{equation}\label{Psi_cy}
  \varPhi = \left(
    \begin{array}{cc}
      \cos\varphi & \sin\varphi \\
      -\sin\varphi & \cos\varphi \\
    \end{array}
  \right),\qquad c_y = \frac{\|y_2\|}{\|y_1\|}
\end{equation}
such that
\begin{equation}\label{Psi_eqn}
  M_{y_2} = c_yM_{y_1}\varPhi.
\end{equation}
Making use of \eqref{Psi_eqn}, the formula in \eqref{eqn_u_2_I} becomes
\begin{equation}\label{rewrite_u_2_I}
\begin{split}
  u_2^I &=\left(
           \begin{array}{cc}
             w_{Re}(0) & w_{Im}(0) \\
           \end{array}
         \right)
  \left(
    \begin{array}{cc}
      y_{2,Re} & -y_{2,Im} \\
      -y_{2,Im} & -y_{2,Re} \\
    \end{array}
  \right)\left(
         \begin{array}{c}
         \cos(\omega_ot) \\
         \sin(\omega_ot) \\
         \end{array}
         \right)\\
  &=c_y\left(
           \begin{array}{cc}
             w_{Re}(0) & w_{Im}(0) \\
           \end{array}
         \right)
  \left(
    \begin{array}{cc}
      y_{1,Re} & -y_{1,Im} \\
      -y_{1,Im} & -y_{1,Re} \\
    \end{array}
  \right)\left(
    \begin{array}{cc}
      \cos\varphi & \sin\varphi \\
      -\sin\varphi & \cos\varphi \\
    \end{array}
  \right)\left(
         \begin{array}{c}
         \cos(\omega_ot) \\
         \sin(\omega_ot) \\
         \end{array}
         \right)\\
  &=c_y\left(
           \begin{array}{cc}
             w_{Re}(0) & w_{Im}(0) \\
           \end{array}
         \right)
  \left(
    \begin{array}{cc}
      y_{1,Re} & -y_{1,Im} \\
      -y_{1,Im} & -y_{1,Re} \\
    \end{array}
  \right)\left(
         \begin{array}{c}
         \cos\varphi\cos(\omega_ot) +\sin\varphi\sin(\omega_ot) \\
         -\sin\varphi\cos(\omega_ot) +\cos\varphi\sin(\omega_ot) \\
         \end{array}
         \right)\\
  &=c_y\left(
           \begin{array}{cc}
             w_{Re}(0) & w_{Im}(0) \\
           \end{array}
         \right)
  \left(
    \begin{array}{cc}
      y_{1,Re} & -y_{1,Im} \\
      -y_{1,Im} & -y_{1,Re} \\
    \end{array}
  \right)\left(
         \begin{array}{c}
         \cos(\omega_ot-\varphi) \\
         \sin(\omega_ot-\varphi)  \\
         \end{array}
         \right)
\end{split}
\end{equation}
Similarly, \eqref{eqn_u_2_II} can be rewritten as
\begin{equation}\label{rewrite_u_2_II}
\begin{split}
  u_2^{II}&=\left(
           \begin{array}{cc}
             f_{Re}(s) & f_{Im}(s) \\
           \end{array}
         \right)
  \left(
    \begin{array}{cc}
      y_{2,Re} & -y_{2,Im} \\
      -y_{2,Im} & -y_{2,Re} \\
    \end{array}
  \right)\left(
         \begin{array}{c}
         \cos(\omega_o(t-s)) \\
         \sin(\omega_o(t-s)) \\
         \end{array}
         \right)\\
  &=c_y\left(
           \begin{array}{cc}
             f_{Re}(s) & f_{Im}(s) \\
           \end{array}
         \right)\left(
    \begin{array}{cc}
      y_{1,Re} & -y_{1,Im} \\
      -y_{1,Im} & -y_{1,Re} \\
    \end{array}
  \right)\left(
         \begin{array}{c}
         \cos(\omega_o(t-s)-\varphi) \\
         \sin(\omega_o(t-s)-\varphi)  \\
         \end{array}
         \right)
\end{split}
\end{equation}
Therefore, in light of \eqref{rewrite_u_2_I} and \eqref{rewrite_u_2_II}, \eqref{u2_form} becomes
\begin{equation}\label{u2_form_new}
\begin{split}
  u_2(t)
   &= 2e^{-d_ot}\Big[y_{2,Re}\Big(w_{Re}(0)\cos(\omega_ot)-w_{Im}(0)\sin(\omega_ot)\Big)\\
   & \qquad\qquad\qquad\qquad\qquad\qquad- y_{2,Im}\Big(w_{Re}(0)\sin(\omega_ot)+w_{Im}(0)\cos(\omega_ot)\Big)\Big]\\
   &\qquad+2\int_0^t e^{-d_o(t-s)}\Big[y_{2,Re}\Big(\cos\big(\omega_o(t-s)\big)f_{Re}(s)-\sin\big(\omega_o(t-s)\big)f_{Im}(s)\Big)\\
   &\qquad\qquad-y_{2,Im}\Big(\cos\big(\omega_o(t-s)\big)f_{Im}(s)+\sin\big(\omega_o(t-s)\big)f_{Re}(s)\Big)
   \Big]ds\\
   &= 2c_ye^{-d_ot}\Big[y_{1,Re}\Big(w_{Re}(0)\cos(\omega_ot-\varphi)-w_{Im}(0)\sin(\omega_ot-\varphi)\Big)\\
   &\qquad\qquad\qquad- y_{1,Im}\Big(w_{Re}(0)\sin(\omega_ot)+w_{Im}(0)\cos(\omega_ot)\Big)\Big]\\
   &\quad +2c_y\int_0^t e^{-d_o(t-s)}\Big[y_{1,Re}\Big(\cos\big(\omega_o(t-s)-\varphi\big)f_{Re}(s)-\sin\big(\omega_o(t-s)-\varphi\big)f_{Im}(s)\Big)\\
   &\qquad\qquad-y_{1,Im}\Big(\cos\big(\omega_o(t-s)-\varphi\big)f_{Im}(s)+\sin\big(\omega_o(t-s)-\varphi\big)f_{Re}(s)\Big)
   \Big]ds.
\end{split}
\end{equation}
Comparing \eqref{u1_form} and \eqref{u2_form_new} leads to
\begin{equation}\label{relation_u2_u1}
\begin{split}
  u_2(t) &= c_ye^{-d_o\varphi/\omega_o}u_1(t-\varphi/\omega_o) \\
  &\quad+ 2c_y\int_{t-\varphi/\omega_o}^t e^{-d_o(t-s)}\bigg[y_{1,Re}\Big(\cos\big(\omega_o(t-s)-\varphi\big)f_{Re}(s)\\
   &\qquad\qquad\qquad\qquad\qquad\qquad\qquad\qquad\qquad-\sin\big(\omega_o(t-s)-\varphi\big)f_{Im}(s)\Big)\\
   &\qquad-y_{1,Im}\Big(\cos\big(\omega_o(t-s)-\varphi\big)f_{Im}(s)+\sin\big(\omega_o(t-s)-\varphi\big)f_{Re}(s)\Big)\bigg]ds
\end{split}
\end{equation}
Note that in \eqref{Define_ab} we defined $\widetilde{a}:=u_1$ and $\widetilde{b}:=u_2$. Thus, plugging \eqref{relation_u2_u1} to the first equation of \eqref{Discharge_Recharge_General2} yields
\begin{equation}\label{Delayed_detail}
\begin{split}
  \frac{\d\widetilde{a}(t)}{\d t}  &= -d_o  \widetilde{a}(t) + \widetilde{c}_{11}\widetilde{a}(t) +\widetilde{c}_{12}\widetilde{b}(t) + \widetilde{g}_{a}(t)\\
    &=-d_o  \widetilde{a}(t) + \widetilde{c}_{11}\widetilde{a}(t) +\widetilde{c}_{12}c_ye^{-d_o\varphi/\omega_o}\widetilde{a}(t-\varphi/\omega_o) + \widetilde{g}_{a}(t)\\
    &\quad + 2c_y\widetilde{c}_{ab}\int_{t-\varphi/\omega_o}^t e^{-d_o(t-s)}\bigg[y_{1,Re}\Big(\cos\big(\omega_o(t-s)-\varphi\big)f_{Re}(s)\\
&\qquad\qquad\qquad\qquad\qquad\qquad\qquad\qquad -\sin\big(\omega_o(t-s)-\varphi\big)f_{Im}(s)\Big)\\
   &\qquad-y_{1,Im}\Big(\cos\big(\omega_o(t-s)-\varphi\big)f_{Im}(s)+\sin\big(\omega_o(t-s)-\varphi\big)f_{Re}(s)\Big)\bigg]\d s
\end{split}
\end{equation}
\end{proof}

\bibliography{references}

\end{document}